\documentclass[11pt]{article}

\usepackage{graphicx}

\usepackage{amsmath,amssymb}
\usepackage{bm}
\usepackage{enumerate}
\usepackage{amsthm}
\usepackage[all,2cell]{xy}
\usepackage{mathrsfs}
\usepackage{mathtools}
\mathtoolsset{showonlyrefs}

\usepackage{tikz-cd}
\usepackage{xcolor}
\DeclareUnicodeCharacter{0301}{*************************************}

\usepackage[colorlinks=true]{hyperref}
\hypersetup{urlcolor=blue, citecolor=red, linkcolor=blue}

\usetikzlibrary{cd}

\usepackage[a4paper, left=25mm, right=25mm, top=30mm, bottom=30mm]{geometry}

\usepackage{marginnote}

\usepackage{imakeidx}
\makeindex[intoc]

\usepackage[nottoc]{tocbibind}
\usepackage[colorinlistoftodos]{todonotes} 

\theoremstyle{plain}
\newtheorem{thm}{Theorem}[section]
\newtheorem{prop}[thm]{Proposition}

\newtheorem{lem}[thm]{Lemma}
\newtheorem{exmpl}[thm]{Example}

\theoremstyle{definition}
\newtheorem{dfn}[thm]{Definition}
\newtheorem{rmrk}[thm]{Remark}

\newcommand\restr[2]{{
  \left.\kern-\nulldelimiterspace 
  #1 
  \right|_{#2} 
}}

\newcommand*{\transp}[2][-3mu]{\ensuremath{\mskip1mu\prescript{\smash{\mathrm t\mkern#1}}{}{\mathstrut#2}}}%

\newcommand{\R}{\mathbb{R}}
\renewcommand{\d}{\mathrm{d}}

\newcommand{\Cinfty}{\mathscr{C}^\infty}
\newcommand{\T}{\mathrm{T}}
\newcommand{\cT}{\mathrm{T}^\ast}
\newcommand{\Id}{\mathrm{Id}}

\newcommand{\Lie}{\mathscr{L}}

\newcommand{\X}{\mathfrak{X}}
\newcommand{\Xh}{\mathfrak{X}_{\rm ham}}

\renewcommand{\L}{\mathcal{L}}
\newcommand{\F}{\mathcal{F}}

\newcommand{\C}{\mathcal{C}}
\newcommand{\D}{\mathcal{D}}
\newcommand{\V}{\mathcal{V}}

\newcommand{\rmC}{\mathrm{C}}
\newcommand{\rmR}{\mathrm{R}}
\newcommand{\rmS}{\mathrm{S}}

\newcommand{\bfX}{\mathbf{X}}
\newcommand{\bfY}{\mathbf{Y}}

\newcommand{\parder}[2]{\frac{\partial #1}{\partial #2}}
\newcommand{\dparder}[2]{\dfrac{\partial #1}{\partial #2}}
\newcommand{\tparder}[2]{\partial #1/\partial #2}
\newcommand{\parderr}[3]{\frac{\partial^2 #1}{\partial #2\partial #3}}
\newcommand{\dparderr}[3]{\dfrac{\partial^2 #1}{\partial #2\partial #3}}

\let\graph\relax
\DeclareMathOperator{\graph}{graph}

\DeclareMathOperator{\Ima}{Im}

\DeclareMathOperator{\Hom}{Hom}
\DeclareMathOperator{\rk}{rank}

\DeclareMathAlphabet{\mathpzc}{OT1}{pzc}{m}{it}
\def\d{\mathrm{d}}

\let\emph\textbf

\title{{\sffamily Nonautonomous $k$-contact field theories}}

\author{{\sffamily
Xavier Rivas%
\thanks{e-mail:
   xavier.rivas@unir.net \ ORCID: 0000-0002-4175-5157}
}
\\[1ex]
\normalsize\itshape\sffamily
Escuela Superior de Ingenier\'{\i}a y Tecnolog\'{\i}a,\\
\normalsize\itshape\sffamily
Universidad Internacional de La Rioja, Logro\~no, Spain.
\\[1ex]
}
\date{{\sffamily \today}}

\begin{document}

\maketitle

\begin{abstract}
    This paper provides a new geometric framework to describe non-conservative field theories with explicit dependence on the space-time coordinates by combining the $k$-cosymplectic and $k$-contact formulations. This geometric framework, the $k$-cocontact geometry, permits to develop a Hamiltonian and Lagrangian formalisms for these field theories. We also compare this new formulation in the autonomous case with the previous $k$-contact formalism. To illustrate the theory, we study the nonlinear damped wave equation with external time-dependent forcing.
\end{abstract}


\noindent\textbf{MSC\,2020 codes:}

{\sl Primary:}
70S05, 
{\sl Secondary:}
53Z05, 
53D10, 
53C15, 

{\setcounter{tocdepth}{2}
\def\baselinestretch{1}
\small
\def\addvspace#1{\vskip 1pt}
\parskip 0pt plus 0.1mm
\tableofcontents
}

\newpage



\section{Introduction}

During the second half of the 20th century, geometric methods have been widely applied to mechanics and field theory with the aim of providing geometric descriptions of a large variety of systems in applied mathematics, physics, engineering, etc. Some of the most frequent geometric structures involved in geometric mechanics and field theory are symplectic, multisymplectic or $k$-symplectic manifolds (see for instance \cite{Abr1978,Arn1989,Awa1992,Car1991,DeLeo1989,DeLeo2015,Kij1973,Kij1979,Lib1987,Rom2011} and references therein). In general, all these geometric methods are applied to Lagrangian and Hamiltonian conservative systems, that is, without damping.

In recent years, there has been a growing interest for non-conservative systems. In particular, contact geometry \cite{Ban2016,Gei2008,Kho2013} has been used to study mechanical systems with dissipation \cite{Bra2017a,Bra2017b,Car2019,DeLeo2019b,Gas2019,Liu2018}. This has many applications in thermodynamics \cite{Bra2018,Sim2020}, quantum mechanics \cite{Cia2018}, circuit theory \cite{Got2016} and control theory \cite{Ram2017} among others \cite{Bra2020,DeLeo2021d,DeLeo2021b,LLLR-2022,GG22,GG22b}. Recently, contact mechanics have been generalized in order to deal with time-dependent contact systems \cite{DeLeo2022,RiTo-2022}. It is worth pointing out that contact geometry allows to study more systems than just dissipative ones \cite{LR-2022}. In the last years, a generalization of both contact and $k$-symplectic structures was devised to describe autonomous field theories with damping \cite{Gas2020,Gas2021,Gra2021} both in the Hamiltonian and Lagrangian formulations.

The main goal of this paper is to extend the $k$-contact formulation to non-autonomous field theories by combining it with $k$-cosymplectic geometry \cite{DeLeo1998,DeLeo2001}. This leads to the definition of a \textsl{$k$-cocontact structure} as a couple of families of $k$ differential one-forms: the first family accounting for the space-time coordinates, and the other one encoding the dampings or dissipations, inspired in the contact formulation. It is worth noting that the number of independent variables of the system coincides with the number of ``dissipation coordinates''. This new geometry enables us to introduce the notion of $k$-cocontact Hamiltonian system as a $k$-cocontact manifold together with a Hamiltonian function. With these elements we can state the $k$-cocontact Hamilton equations, which indeed add dissipation terms to the well-known Hamiltonian field equations \cite{DeLeo2015}.

In addition we also generalize the Lagrangian formulation of field theories to consider non-autonomous non-conservative ones. In this new formalism, the phase bundle is $M = \R^k\times\bigoplus^k\T Q\times\R^k$, where the direct sum has to be understood as a fibered sum of vector bundles, with adapted coordinates $(t^\alpha,q^i,v_\alpha^i,z^\alpha)$. Then, given a Lagrangian function $L:M\to\R$, we define a family $\eta^\alpha_L$ of one-forms which, when $L$ is regular, constitute along with the forms $\d t^\alpha$ a $k$-cocontact structure on $M$. Then, the $k$-cocontact Lagrangian field equations are the $k$-cocontact Hamiltonian field equations for the Lagrangian energy. When written in coordinates, they are the Euler--Lagrange equations with some additional damping terms.

We also compare the $k$-cocontact formalism introduced in this work in the autonomous case with the previous $k$-contact formalism and see that they are partially equivalent, in the same way as autonomous $k$-cosymplectic systems are closely related to $k$-symplectic systems \cite{DeLeo2015}. Finally, we apply this formalism to the nonlinear damped wave equation with a time-dependent external force, both in the Hamiltonian and Lagrangian formulations.

The structure of the paper is as follows. In Section \ref{sec:2} we provide a review of the $k$-contact formalism for non-conservative autonomous field theories. In particular, we provide the main results on $k$-contact geometry and a brief description of the Hamiltonian and Lagrangian formalisms. Section \ref{sec:3} is devoted to the introduction of the notion of $k$-cocontact structure and study its geometry. More precisely, we prove the existence of two families of Reeb vector fields and the existence of two types of special sets of coordinates: adapted coordinates and, by adding an extra hypothesis, Darboux coordinates.

In Section \ref{sec:4} we develop a Hamiltonian formalism for non-autonomous field theories with damping, generalizing the De Donder--Weyl formulation for field theories. We provide field equations both for $k$-vector fields and integral sections, and we prove the existence (and not uniqueness) of solutions to these equations. We begin Section \ref{sec:5} by describing the geometry of the phase bundle of $k$-cocontact Lagrangian field theories. We also present the field equations, generalizing the Euler--Lagrange equations and give the conditions for a Lagrangian function to be regular, that is to yield a $k$-cocontact structure. Finally, we study a particularly interesting type of Lagrangian functions: the Lagrangians with holonomic damping term.

Section \ref{sec:6} is devoted to compare the $k$-contact formalism introduced in \cite{Gas2020} with the $k$-cocontact setting presented in this work in the autonomous case. In order to illustrate the geometric formalism introduced in previous sections, in Section \ref{sec:7} study the example of a nonlinear damped wave equation, describing all the geometric objects involved, both in the Lagrangian and Hamiltonian formulations.

Unless otherwise stated, all maps are assumed to be $\Cinfty$ and all manifolds are smooth, connected and second countable.
Sum over crossed repeated indices is understood. The direct sum of two vector bundles over the same base space is to be understood as the Whitney sum of vector bundles.

\section{Review on \texorpdfstring{$k$}--contact systems}\label{sec:2}

In this first section we review the $k$-contact formalism for non-conservative field theories. In first place we introduce the geometric framework: $k$-contact structures. Then, the Hamiltonian \cite{Gas2020} and the Lagrangian \cite{Gas2021} formalisms are presented.

\subsection{\texorpdfstring{$k$}--contact manifolds}

Consider an $m$-dimensional manifold $M$. A \textsl{generalized distribution} on $M$ is a subset $D\subset \T M$ such that $D_x\subset\T_xM$ is a vector subspace for every $x\in M$. A distribution $D$ is said to be \textsl{smooth} if it can be locally spanned by a family of vector fields, and \textsl{regular} if it is smooth and of locally constant rank. A \textsl{codistribution} on $M$ is a subset $C\subset\cT M$ such that $C_x\subset\cT_xM$ is a vector subspace for every $x\in M$.

Given a distribution $D$, the anihilator $D^\circ$ of $D$ is a codistribution. If $D$ is not regular, $D^\circ$ may not be smooth. Using the usual identification $E^{\ast\ast} = E$ of finite-dimensional linear algebra, it is clear that $(D^\circ)^\circ = D$.

A differential one-form $\eta\in\Omega^1(M)$ generates a smooth codistribution, denoted by $\langle\eta\rangle\subset\cT M$. This codistribution has rank 1 at every point where $\eta$ does not vanish. Its anihilator is a distribution $\langle\eta\rangle^\circ\subset\T M$ that can be described as the kernel of the linear morphism $\widehat\eta\colon\T M\to M\times\R$ defined by $\eta$. This codistribution has corank 1 at every point where $\eta$ does not vanish.

In the same way, every two-form $\omega\in\Omega^2(M)$ induces a linear morphism $\widehat\omega\colon\T M\to\cT M$ defined as $\widehat\omega(v) = i(v)\omega$. The kernel of this morphism $\widehat\omega$ is a distribution $\ker\widehat\omega\subset\T M$. Notice that the rank of $\widehat\omega$ is even.

Given a family of $k$ differential one-forms $\eta^1,\dotsc,\eta^k\in\Omega^1(M)$, we will denote
\begin{itemize}
    \item $\C^\rmC = \langle\eta^1,\dotsc,\eta^k\rangle\subset\cT M$,
    \item $\D^\rmC = \left(\C^\rmC\right)^\circ = \ker\widehat{\eta^1}\cap\dotsb\cap\ker\widehat{\eta^k}\subset\T M$,
    \item $\D^\rmR = \ker\widehat{\d\eta^1}\cap\dotsb\cap\ker\widehat{\d\eta^k}\subset\T M$,
    \item $\C^\rmR = \left( \D^\rmR \right)^\circ\subset\cT M$.
\end{itemize}
With the preceding notations, a \textsl{$k$-contact structure} on a manifold $M$ is a family of $k$ differential one-forms $\eta^1,\dotsc,\eta^k\in\Omega^1(M)$ such that $\D^\rmC\subset\T M$ is a regular distribution of corank $k$, $\D^\rmR\subset\T M$ is a regular distribution of rank $k$ and $\D^\rmC\cap\D^\rmR = \{0\}$. We call $\C^\rmC$ the \textsl{contact codistribution}, $\D^\rmC$ the \textsl{contact distribution}, $\D^\rmR$ the \textsl{Reeb distribution} and $\C^\rmR$ the \textsl{Reeb codistribution}. A manifold $M$ endowed with a $k$-contact structure $\eta^1,\dotsc,\eta^k\in\Omega^1(M)$ is a \textsl{$k$-contact manifold}.

\begin{rmrk}
    In the particular case $k=1$, a 1-contact structure is given by a one-form $\eta$. In this case, we recover the notion of contact manifold \cite{DeLeo2019b,Gas2019}.
\end{rmrk}

Given a $k$-contact manifold $(M,\eta^\alpha)$, the Reeb distribution $\D^\rmR$ is involutive, and therefore integrable, and there exists a unique family of $k$ vector fields $R_\alpha\in\X(M)$, called \textsl{Reeb vector fields} of $M$, such that $i(R_\alpha)\eta^\beta = \delta_\alpha^\beta$ and $i(R_\alpha)\d\eta^\beta = 0$. The Reeb vector fields commute and span the Reeb distribution $\D^\rmR = \langle R_1,\dotsc,R_k \rangle$.

\begin{exmpl}[Canonical $k$-contact structure]\label{ex:canonical-k-contact-structure}
    Consider $k\geq 1$ and let $Q$ be a smooth manifold. The manifold product $M = \oplus^k\cT Q\times\R^k$ has a canonical contact structure given by the one-forms $\eta^1,\dotsc,\eta^k\in\Omega^1(M)$ defined as
    $$ \eta^\alpha = \d z^\alpha - \theta^\alpha\,, $$
    where $(z^1,\dotsc,z^k)$ are the canonical coordinates of $\R^k$ and $\theta^\alpha$ is the pull-back of the Liouville one-form $\theta$ of the cotangent bundle $\cT Q$ with respect to the projection $\mathrm{pr}^\alpha:M\to\cT Q$ to the $\alpha$-th component. Take coordinates $(q^i)$ on $Q$. Then, $M$ has natural coordinates $(q^i,p_i^\alpha,z^\alpha)$. Using these coordinates, we have
    $$ \eta^\alpha = \d z^\alpha - p_i^\alpha\d q^i\,,\qquad\D^\rmR = \left\langle\parder{}{z^1},\dotsc,\parder{}{z^k}\right\rangle\,,\qquad R_\alpha = \parder{}{z^\alpha}\,. $$
\end{exmpl}

\begin{exmpl}[Contactification of a $k$-symplectic manifold]\label{ex:contactification-k-symplectic-manifold}
    Consider a $k$-symplec\-tic manifold $(P,\omega^\alpha)$ such that $\omega^\alpha = -\d\theta^\alpha$ and the product manifold $M = P\times\R^k$. Let $(z^\alpha)$ be the cartesian coordinates of $\R^k$ and denote also by $\theta^\alpha$ the pull-back of $\theta^\alpha$ to the product manifold $M$. Consider the one-forms $\eta^\alpha = \d z^\alpha - \theta^\alpha\in\Omega^1(M)$.
    
    Then, $(M,\eta^\alpha)$ is a $k$-contact manifold because $\C^\rmC = \langle\eta^1,\dotsc,\eta^k\rangle$ has rank $k$, $\d\eta^\alpha = -\d\theta^\alpha$, and $\D^\rmR = \bigcap_\alpha\ker\widehat{\d\theta^\alpha} = \langle\tparder{}{z^1},\dotsc,\tparder{}{z^k}\rangle$ has rank $k$ since $(P,\omega^\alpha)$ is $k$-symplectic, and the last condition is immediate.
    
    Notice that the so-called canonical $k$-contact structure described in Example \ref{ex:canonical-k-contact-structure} is just the contactification of the $k$-symplectic manifold $P = \oplus^k\cT Q$.
    
\end{exmpl}

\begin{thm}[$k$-contact Darboux Theorem]\label{thm:k-contact-Darboux}
    Consider a $k$-contact manifold $(M,\eta^\alpha)$ of dimension $\dim M = n + kn + k$ such that there exists an integrable subdistribution $\V$ of $\D^\rmC$ with $\rk\V = nk$. Then, around every point of $M$ there exists a local chart $(U,q^i,p_i^\alpha,z^\alpha)$, $1\leq\alpha\leq k$, where $1\leq i\leq n$, such that
    $$ \restr{\eta^\alpha}{U} = \d z^\alpha - p_i^\alpha\d q^i\,,\quad \restr{\D^\rmR}{U} = \left\langle R_\alpha = \parder{}{z^\alpha}\right\rangle\,,\quad\restr{\V}{U} = \left\langle\parder{}{p_i^\alpha}\right\rangle\,. $$
    These coordinates are called \textsl{Darboux coordinates} of the $k$-contact manifold $(M,\eta^\alpha)$.
\end{thm}

\subsection{Hamiltonian formalism for \texorpdfstring{$k$}--contact systems}
\label{sub:k-contact-Hamiltonian-systems}

The geometric setting introduced in the previous section allows us to introduce the notion of $k$-contact Hamiltonian system \cite{Gas2020}.

\begin{dfn}\label{dfn:k-contact-Hamiltonian-system}
    A \textsl{$k$-contact Hamiltonian system} is a family $(M, \eta^\alpha,h)$, where $(M,\eta^\alpha)$ is a $k$-contact manifold and $h\in\Cinfty(M)$ is called a \textsl{Hamiltonian function}. Consider a map $\psi\colon D\subset\R^k\to M$. The \textsl{$k$-contact Hamilton--De Donder--Weyl equations for the map $\psi$} are
    \begin{equation}\label{eq:k-contact-HDW}
        \begin{dcases}
            i(\psi_\alpha')\d\eta^\alpha = \left( \d h - (\Lie_{R_\alpha}h)\eta^\alpha \right)\circ\psi\,,\\
            i(\psi_\alpha')\eta^\alpha = -h\circ\psi\,.
        \end{dcases}
    \end{equation}
\end{dfn}

In Darboux coordinates, if the map $\psi$ has local expression $\psi(t) = (q^i(t),p_i^\alpha(t),z^\alpha(t))$, equations \eqref{eq:k-contact-HDW} read
\begin{equation}\label{eq:k-contact-HDW-Darboux-coordinates}
    \begin{dcases}
        \parder{q^i}{t^\alpha} = \parder{h}{p_i^\alpha}\circ\psi\,,\\
        \parder{p^\alpha_i}{t^\alpha} = -\left( \parder{h}{q^i} + p_i^\alpha\parder{h}{z^\alpha} \right)\circ\psi\,,\\
        \parder{z^\alpha}{t^\alpha} = \left( p_i^\alpha\parder{h}{p_i^\alpha} - h \right)\circ\psi\,.
    \end{dcases}
\end{equation}

\begin{dfn}
    Consider a $k$-contact Hamiltonian system $(M,\eta^\alpha,h)$ and a $k$-vector field $\bfX = (X_\alpha)\in\X^k(M)$. The \textsl{$k$-contact Hamilton--De Donder--Weyl equations for the $k$-vector field} $\bfX$ are
    \begin{equation}\label{eq:k-contact-HDW-fields}
        \begin{dcases}
            i(X_\alpha)\d\eta^\alpha = \d h - (\Lie_{R_\alpha}h)\eta^\alpha\,,\\
            i(X_\alpha)\eta^\alpha = -h\,.
        \end{dcases}
    \end{equation}
    A $k$-vector field solution to these equations is a \textsl{$k$-contact Hamiltonian $k$-vector field}.
\end{dfn}

\begin{prop}\label{prop:k-contact-HDW-have-solutions}
    The $k$-contact Hamilton--De Donder--Weyl equations \eqref{eq:k-contact-HDW-fields} admit solutions. They are not unique if $k > 1$.
\end{prop}

Consider a $k$-vector field $\bfX = (X_1,\dotsc,X_k)\in\X^k(M)$ with local expression in Darboux coordinates
$$ X_\alpha = (X_\alpha)^i\parder{}{q^i} + (X_\alpha)^\beta_i\parder{}{p_i^\beta} + (X_\alpha)^\beta\parder{}{z^\beta}\,. $$
Then, equation \eqref{eq:k-contact-HDW-fields} yields the conditions
\begin{equation}\label{eq:k-contact-HDW-fields-Darboux-coordinates}
    \begin{dcases}
        (X_\alpha)^i = \parder{h}{p_i^\alpha}\,,\\
        (X_\alpha)^\alpha_i = -\left( \parder{h}{q^i} + p_i^\alpha\parder{h}{z^\alpha} \right)\,,\\
        (X_\alpha)^\alpha = p_i^\alpha\parder{h}{p_i^\alpha} - h\,.
    \end{dcases}
\end{equation}

\begin{prop}\label{prop:k-contact-equiv-fields-sections}
    Consider an integrable $k$-vector field $\bfX\in\X^k(M)$. Then, every integral section $\psi\colon D\subset\R^k\to M$ of $\bfX$ satisfies the $k$-contact Hamilton--De Donder--Weyl equations \eqref{eq:k-contact-HDW} if, and only if, $\bfX$ is a solution to \eqref{eq:k-contact-HDW-fields}.
\end{prop}


\begin{prop}
    The $k$-contact Hamilton--De Donder--Weyl equations \eqref{eq:k-contact-HDW-fields} are equivalent to
    \begin{equation}\label{eq:k-contact-HDW-alternative}
        \begin{dcases}
            \Lie_{X_\alpha}\eta^\alpha = -(\Lie_{R_\alpha}h)\eta^\alpha\,,\\
            i(X_\alpha)\eta^\alpha = -h\,.
        \end{dcases}
    \end{equation}
\end{prop}

\subsection{Lagrangian formalism for \texorpdfstring{$k$}--contact systems}

The Hamiltonian formalism presented in the previous section has a Lagrangian counterpart. Consider the phase bundle $\bigoplus^k\T Q\times\R^k$ endowed with adapted coordinates $(q^i, v^i_\alpha, z^\alpha)$ with the usual canonical structures: the Liouville vector field $\Delta = v_\alpha^i\parder{}{v_\alpha^i}$ and the canonical $k$-tangent structure $J^\alpha = \parder{}{v_\alpha^i}\otimes\d q^i$ (see \cite{Gas2021} for details). A $k$-vector field $\bfX = (X_\alpha)\in\X^k(\bigoplus^k\T Q\times\R^k)$ is a \textsl{second-order partial differential equation} (or {\sc sopde}) is $J^\alpha(X_\alpha) = \Delta$.

Given a Lagrangian function $L\colon\bigoplus^k\T Q\times\R\to\R$, the \textsl{Lagrangian energy} associated to the Lagrangian $L$ is the function $E_L = \Delta(L) - L$, and the \textsl{contact one-forms} $\eta^\alpha_L\in\Omega^1(\bigoplus^k\T Q\times\R^k)$ associated to $L$ are given by $\eta_L^\alpha = \d z^\alpha - \theta_L^\alpha$, where $\theta^\alpha_L = \transp{J^\alpha}\circ\d L$.

The Lagrangian $L$ is \textsl{regular}, namely $\parderr{L}{v_\alpha^i}{v_\beta^j}$ is non-degenerate, if and only if the contact one-forms $\eta^\alpha_L$ define a $k$-contact structure on $\bigoplus^k\T Q\times\R^k$. Thus, we can consider the $k$-contact Hamiltonian system $(\bigoplus^k\T Q\times\R^k,\eta_L^\alpha,E_L)$, whose corresponding field equations read
\begin{equation}\label{eq:k-contact-Euler-Lagrange-section-coordinates}
	\begin{dcases}
		\parder{}{t^\alpha}\left( \parder{L}{v^i_\alpha} \circ\psi\right) = \left( \parder{L}{q^i} + \parder{L}{s^\alpha}\parder{L}{v^i_\alpha} \right)\circ\psi\,,\\
		\parder{(s^\alpha\circ\psi)}{t^\alpha} = L\circ\psi\,,
	\end{dcases}
\end{equation}
and are called \textsl{$k$-contact Euler--Lagrange equations} (for more details on the $k$-contact Lagrangian formulation, see \cite{Gas2021,PhDThesisXRG}).

\section{\texorpdfstring{$k$}--cocontact geometry}\label{sec:3}

Let $\tau^1,\dotsc,\tau^k\in\Omega^1(M)$ be a family of closed one-forms on $M$ and let $\eta^1,\dotsc,\eta^k\in\Omega^1(M)$ be a family of one-forms on $M$. We will use the following notations:
\begin{itemize}
    \item $\C^\rmC = \langle\eta^1,\dotsc,\eta^k\rangle\subset\cT M$,
    \item $\D^\rmC = \left(\C^\rmC\right)^\circ = \ker\widehat{\eta^1}\cap\dotsb\cap\ker\widehat{\eta^k}\subset\T M$,
    \item $\D^\rmR = \ker\widehat{\d\eta^1}\cap\dotsb\cap\ker\widehat{\d\eta^k}\subset\T M$,
    \item $\C^\rmR = \left(\D^\rmR\right)^\circ\subset\cT M$,
    \item $\C^\rmS = \langle\tau^1,\dotsc,\tau^k\rangle\subset\cT M$,
    \item $\D^\rmS = \left(\C^\rmS\right)^\circ = \ker\widehat{\tau^1}\cap\dotsb\cap\ker\widehat{\tau^k}\subset\T M$.
\end{itemize}
With these notations, we can define the notion of $k$-cocontact structure:
\begin{dfn}\label{dfn:k-cocontact-manifold}
    A \textsl{$k$-cocontact structure} on a manifold $M$ is a family of $k$ closed differential one-forms $\tau^1,\dotsc,\tau^k\in\Omega^1(M)$ and a family of $k$ differential one-forms $\eta^1,\dotsc,\eta^k\in\Omega^1(M)$ such that, with the preceding notations,
    \begin{enumerate}[{\rm(1)}]
        \item $\D^\rmC\subset\T M$ is a regular distribution of corank $k$,
        \item $\D^\rmS\subset\T M$ is a regular distribution of corank $k$,
        \item $\D^\rmR\subset\T M$ is a regular distribution of rank $2k$,
        \item $\D^\rmC \cap\D^\rmS$ is a regular distribution of corank $2k$, $\D^\rmC \cap\D^\rmR$ is a regular distribution of rank $k$, and $\D^\rmS \cap\D^\rmR$ is a regular distribution of rank $k$,
        \item $\D^\rmC\cap\D^\rmR\cap\D^\rmS = \{0\}$.
    \end{enumerate}
    We call $\C^\rmC$ the \textsl{contact codistribution}, $\D^\rmC$ the \textsl{contact distribution}, $\D^\rmR$ the \textsl{Reeb distribution}, $\C^\rmR$ the \textsl{Reeb codistribution}, $\C^\rmS$ the \textsl{space-time codistribution} and $\D^\rmS$ the \textsl{space-time distribution}.
    
    A manifold $M$ endowed with a $k$-cocontact structure $\tau^1,\dotsc,\tau^k,\eta^1,\dotsc,\eta^k\in\Omega^1(M)$ is a \textsl{$k$-cocontact manifold}.
\end{dfn}

Notice that the condition $\D^\rmC\cap\D^\rmR\cap\D^\rmS = \{0\}$ implies that
$$ \cT M = \C^\rmC \oplus \C^\rmR \oplus \C^\rmS \,. $$

\begin{rmrk}
    In the particular case $k=1$, a 1-cocontact structure is given by two one-forms $\tau,\eta$, with $\d\tau=0$. The conditions in Definition \ref{dfn:k-cocontact-manifold} mean the following: (1) $\eta\neq 0$ everywhere, (2) $\tau\neq 0$ everywhere, (4) $\tau\wedge\eta\neq 0$, (5) $\ker\widehat{\tau}\cap\ker\widehat{\eta}\cap\ker\widehat{\d\eta} = \{0\}$, which implies that $\ker\widehat{\d\eta}$ has rank 0, 1 or 2, and (3) implies that $\ker\widehat{\d\eta}$ has rank 2. Thus, a 1-cocontact structure coincides with the cocontact structure introduced in \cite{DeLeo2022} to describe time-dependent contact mechanics. 
\end{rmrk}

\begin{lem}
    The Reeb distribution $\D^\rmR$ and the space-time distribution $\D^\rmS$ are involutive, and therefore integrable.
\end{lem}
\begin{proof}
    Given $X,Y$ two sections of $\D^\rmR$ and applying the relation
    $$ i_{[X,Y]} = \Lie_Xi_Y - i_Y\Lie_X = \d i_Xi_Y + i_X\d i_Y - i_Y\d i_X - i_Yi_X\d $$
    to the closed two-form $\d\eta^\alpha$, the result is zero. In the same way, one can check that $\D^\rmS$ is also involutive.
\end{proof}

As a consequence, the distribution $\D^\rmR\cap\D^\rmS$ is also involutive, and therefore integrable. Moreover, the distribution $\D^\rmR\cap\D^\rmC$ is also involutive and integrable. The following theorem characterizes a family of vector fields spanning the Reeb distribution $\D^\rmR$.

\begin{thm}\label{thm:Reeb-vector-fields}
    Let $(M,\tau^\alpha,\eta^\alpha)$ be a $k$-cocontact manifold. Then, there exist a unique family $R^t_1,\dotsc,R^t_k,R^z_1,\dotsc,R^z_k\in\X(M)$ such that
    \begin{gather}
        i(R^t_\alpha)\d\eta^\beta = 0\,,\qquad i(R^t_\alpha)\eta^\beta = 0\,,\qquad i(R^t_\alpha)\tau^\beta = \delta_\alpha^\beta\,,\\
        i(R^z_\alpha)\d\eta^\beta = 0\,,\qquad i(R^z_\alpha)\eta^\beta = \delta_\alpha^\beta\,,\qquad i(R^z_\alpha)\tau^\beta = 0\,.
    \end{gather}
    The vector fields $R^t_\alpha$ are called \textsl{space-time Reeb vector fields}. The vector fields $R^z_\alpha$ are called \textsl{contact Reeb vector fields}.

    In addition, the Reeb vector fields commute and span the Reeb distribution introduced in Definition \ref{dfn:k-cocontact-manifold}:
    $$ \D^\rmR = \langle R^t_1,\dotsc,R^t_k,R^z_1,\dotsc,R^z_k\rangle\,, $$
    thus motivating its name.
\end{thm}
\begin{proof}
    Consider $\cT M = \C^\rmC \oplus \C^\rmR\oplus\C^\rmS$. The family of one-forms $\{\eta^\beta\}$ is a global frame of the contact codistribution $\C^\rmC$ and the family of one-forms $\{\tau^\beta\}$ is a global frame of the space-time codistribution $\C^\rmS$. We can find a global frame $\bar\eta^\mu$ of the Reeb codistribution $\C^\rmR$ so that $(\eta^\beta, \bar\eta^\mu,\tau^\beta)$ is a global frame of $\cT M$. Let $(R^z_\alpha, R_\nu,R_\alpha^t)$ be the corresponding dual frame of $\T M$, where the vector fields $R_\alpha^z$ and $R_\alpha^t$ are uniquely determined by the conditions
    \begin{gather}
        \langle \eta^\beta, R_\alpha^z\rangle = \delta_\alpha^\beta\,,\qquad \langle \bar\eta^\mu, R_\alpha^z\rangle = 0\,,\qquad \langle \tau^\beta, R_\alpha^z\rangle = 0\,,\\
        \langle \eta^\beta, R_\alpha^t\rangle = 0\,,\qquad \langle \bar\eta^\mu, R_\alpha^t\rangle = 0\,,\qquad \langle \tau^\beta, R_\alpha^t\rangle = \delta_\alpha^\beta\,.
    \end{gather}
    Notice that the relations involving the $\bar\eta^\mu$ do not depend on the choice of the one-forms $\bar\eta^\mu$, this means that the vector fields $R_\alpha^z$ and $R_\alpha^t$ are sections of the Reeb distribution $(\C^\rmR)^\circ = \D^\rmR$. This amounts to $i(R_\alpha^z)\d\eta^\beta = 0$ and $i(R_\alpha^t)\d\eta^\beta = 0$ for every $\alpha = 1,\dotsc,k$. Since the one-forms $\eta^\beta$ and $\tau^\beta$ are globally defined, so are the vector fields $R_\alpha^z$ and $R_\alpha^t$.

    To prove that the Reeb vector fields $R_\alpha^z,R_\beta^t$ commute, notice that
    $$ i_{[X,Y]}\eta^\gamma = 0\,,\qquad i_{[X,Y]}\d\eta^\gamma = 0\,,\qquad i_{[X,Y]}\tau^\gamma = 0\,, $$
    for every $X,Y\in\langle R_\alpha^z,R_\beta^t\rangle$, which is a consequence of their definition.
\end{proof}

The following proposition proves the existence of a special set of coordinates, the so-called adapted coordinates.

\begin{prop}\label{prop:adapted-coordinates}
    Consider a $k$-cocontact manifold $(M,\tau^\alpha,\eta^\alpha)$. Then, around every point in $M$, there exist local coordinates $(t^\alpha,x^I,z^\alpha)$ such that
    $$ R_\alpha^t = \parder{}{t^\alpha}\,,\qquad \tau^\alpha = \d t^\alpha\,,\qquad R_\alpha^z = \parder{}{z^\alpha}\,,\qquad \eta^\alpha = \d z^\alpha - f_I^\alpha(x^J)\d x^I\,, $$
    where the functions $f_I^\alpha$ only depend on the coordinates $x^I$. These coordinates are called \textsl{adapted coordinates}.
\end{prop}
\begin{proof}
    Since the Reeb vector fields commute, there exists a set of local coordinates $(t^\alpha,x^I,z^\alpha)$ simultaneously straightening out the Reeb vector fields (see \cite[p.234]{Lee2013} for details):
    $$ R_\alpha^t = \parder{}{t^\alpha}\,,\qquad R_\alpha^z = \parder{}{z^\alpha}\,. $$
    Let us write the forms $\tau^\beta$ and $\eta^\beta$ using these coordinates. The conditions $i(R_\alpha^t)\eta^\beta = 0$ and $i(R_\alpha^z)\eta^\beta = \delta_\alpha^\beta$ imply that $\eta^\beta = \d z^\beta - f_I^\beta(t^\alpha,x^J,z^\alpha)\d x^I$. On the other hand, we have that $\d\eta^\beta = \d x^I\wedge\d f_I^\beta$. In this case, the conditions $i(R_\alpha^t)\d\eta^\beta = 0$ and $i(R_\alpha^z)\d\eta^\beta = 0$ imply that $\tparder{f_I^\beta}{t^\alpha} = 0$ and that $\tparder{f_I^\beta}{z^\alpha} = 0$, and thus
    $$ \eta^\beta = \d z^\beta - f_I^\beta(x^J)\d x^I\,. $$
    Repeating this process for the forms $\tau^\beta$, taking into account that $\d\tau^\beta = 0$ and redefining the coordinates $t^\alpha$, we obtain the desired result.
\end{proof}

\begin{exmpl}[Canonical $k$-cocontact structure]\label{ex:canonical -k-cocontact-structure}
    Let $Q$ be a smooth $n$-dimensional manifold with coordinates $(q^i)$ and let $k\geq 1$. Consider the product manifold $M = \R^k\times\bigoplus^k\cT Q\times\R^k$ endowed with natural coordinates $(t^\alpha; q^i, p_i^\alpha; z^\alpha)$. We have the canonical projections
    \begin{center}
        \begin{tikzcd}
            \R & \R^k\times\bigoplus^k\cT Q\times\R^k \arrow[l, swap, "\pi_1^\alpha", ] \arrow[r, "\pi_3^\alpha"] \arrow[dd, "\pi_2"] \arrow[ddd, bend right=50, swap, "\pi_0"] \arrow[ddr, "\pi_2^\alpha"] & \R \\
            &  && \\
            & \bigoplus^k\cT Q \arrow[r, swap, "\pi^\alpha"] & \cT Q \\
            & \R^k\times Q\times\R^k 
        \end{tikzcd}
    \end{center}
    Let $\theta$ be the Liouville one-form on $\cT Q$ with local expression in natural coordinates $\theta = p_i \d q^i$. Then, the family $(\tau^\alpha,\eta^\alpha)$ where $\tau^\alpha = \pi_1^{\alpha\,*}(\d t)$ with $t$ the canonical coordinate of $\R$ and $\eta^\alpha = \d z^\alpha - \pi_2^{\alpha\,*}\theta$, is a $k$-cocontact structure on $M$. In natural coordinates,
    $$ \tau^\alpha = \d t^\alpha\,,\qquad \eta^\alpha = \d z^\alpha - p_i^\alpha\d q^i\,. $$
    Thus, the Reeb vector fields are $R_\alpha^t = \tparder{}{t^\alpha}$ and $R_\alpha^z = \tparder{}{z^\alpha}$.
\end{exmpl}


The following theorem is an upgrade of Proposition 
\ref{prop:adapted-coordinates} and states the existence of Darboux-like coordinates in a $k$-cocontact manifold provided the existence of a certain subdistribution $\mathcal{V}\subset\D^\rmC$.

\begin{thm}[Darboux theorem for $k$-cocontact manifolds]
    Let $(M,\tau^\alpha,\eta^\alpha)$ be a $k$-cocontact manifold with dimension $\dim M = k + n + kn + k$ such that there exists an integrable subdistribution $\mathcal{V}\subset\D^C$ with $\rk \mathcal{V} = nk$. Then, around every point of $M$ there exist local coordinates $(t^\alpha,q^i,p_i^\alpha,z^\alpha)$, where $1\leq\alpha\leq k$ and $1\leq i\leq n$, such that, locally,
    $$ \tau^\alpha = \d t^\alpha\,,\qquad \eta^\alpha = \d z^\alpha - p_i^\alpha\d q^i\,. $$
    Using these coordinates,
    $$ \D^\rmR = \left \langle R_\alpha^t = \parder{}{t^\alpha}\,,\ R_\alpha^z = \parder{}{z^\alpha} \right\rangle\,,\qquad \mathcal{V} = \left\langle \parder{}{p_i^\alpha} \right\rangle\,. $$
    These coordinates are called \textsl{Darboux coordinates} of the $k$-cocontact manifold $(M,\tau^\alpha,\eta^\alpha)$.
\end{thm}
\begin{proof}
    By Proposition \ref{prop:adapted-coordinates}, there exist local coordinates $(t^\alpha,x^I,z^\alpha)$ such that
    $$ R_\alpha^t = \parder{}{t^\alpha}\,,\qquad \tau^\alpha = \d t^\alpha\,,\qquad R_\alpha^z = \parder{}{z^\alpha}\,,\qquad \eta^\alpha = \d z^\alpha - f_I^\alpha(x^J)\d x^I\,. $$
    Since the distribution $\D^\rmC\cap\D^\rmR = \left\langle R_\alpha^t = \dparder{}{t^\alpha} \right\rangle$ is involutive, and therefore integrable, we can consider (at least locally) the quotient manifold $\widetilde M = M/(\D^\rmC\cap\D^\rmR)$, with the projection $\rho\colon M \to \widetilde M$ and local coordinates $(x^I,z^\alpha)$.

    The one-forms $\eta^\alpha$, the vector fields $R_\alpha^z$ and the distribution $\mathcal{V}$ can be projected to $\widetilde M$ and the distribution $\widetilde{\D^\rmC}$ induced by $\D^\rmC$ is $\widetilde{\D^\rmC} = \left\langle R_\alpha^z \right\rangle$.

    It is easy to check that the manifold $(\widetilde M,\widetilde\eta^\alpha)$, where $\widetilde\eta^\alpha$ are the projections of $\eta^\alpha$ to $\widetilde M$, is a $k$-contact manifold. Since the projected distribution $\widetilde{\mathcal{V}}$ has rank $nk$, by Theorem \ref{thm:k-contact-Darboux}, around every point there exists a local chart $(\widetilde U; \widetilde q^{\,i},\widetilde p_i^{\,\alpha},\widetilde z^{\,\alpha})$ in $\widetilde M$ such that
    $$ \widetilde\eta^{\,\alpha} = \d \widetilde z^{\,\alpha} - \widetilde p_i^{\,\alpha}\d\widetilde q^{\;i}\,,\quad \widetilde{\mathcal{V}} = \left\langle \parder{}{\widetilde p_i^{\,\alpha}} \right\rangle\,. $$
    With all this in mind, in $U = \rho^{-1}(\widetilde U)\subset M$, we can take coordinates $(t^\alpha,x^I,z^\alpha) = (t^\alpha,q^i,p_i^\alpha,z^\alpha)$, with $q^i = \widetilde q^{\;i}\circ\rho$, $p_i^\alpha = \widetilde p_i^{\,\alpha}\circ\rho$ and $z^\alpha = \widetilde z^{\,\alpha}\circ\rho$ fulfilling the conditions of the theorem.
\end{proof}

Taking into account the previous theorem, we can consider the manifold introduced in Example \ref{ex:canonical -k-cocontact-structure} as the canonical model for $k$-cocontact structures.

\section{Hamiltonian formalism}\label{sec:4}

This section introduces the notion of $k$-cocontact Hamiltonian system and its Hamilton--De Donder--Weyl equations. The existence of solutions to these equations is proved. We provide local expressions of the Hamilton--De Donder--Weyl equations for maps and $k$-vector fields in both adapted and Darboux coordinates.

\begin{dfn}
    A \textsl{$k$-cocontact Hamiltonian system} is a tuple $(M,\tau^\alpha,\eta^\alpha,h)$, where $(\tau^\alpha,\eta^\alpha)$ is a $k$-cocontact structure on the manifold $M$ and $h\colon M\to\R$ is a \textsl{Hamiltonian function}. Given a map $\psi\colon D\subset\R^k\to M$, the \textsl{$k$-cocontact Hamilton--De Donder--Weyl equations for the map $\psi$} are
    \begin{equation}\label{eq:HDW-map}
        \begin{dcases}
            i(\psi_\alpha')\d\eta^\alpha = \left(\d h - (\Lie_{R_\alpha^t}h)\tau^\alpha - (\Lie_{R_\alpha^z}h)\eta^\alpha\right)\circ\psi\,,\\
            i(\psi_\alpha')\eta^\alpha = -h\circ\psi\,,\\
            i(\psi_\alpha')\tau^\beta = \delta_\alpha^\beta\,.
        \end{dcases}
    \end{equation}
\end{dfn}

Now we are going to look at the expression in coordinates of the Hamilton--De Donder--Weyl equations \eqref{eq:HDW-map}. 

Consider first the adapted coordinates $(t^\alpha,x^I,z^\alpha)$. In these coordinates,
$$ R_\alpha^t = \parder{}{t^\alpha}\,,\quad \tau^\alpha = \d t^\alpha\,,\quad R_\alpha^z = \parder{}{z^\alpha}\,,\quad \eta^\alpha = \d z^\alpha - f_I^\alpha(x^J)\d x^I\,,\quad \d\eta^\alpha = \frac{1}{2}\omega^\alpha_{IJ}\d x^I\wedge\d x^J\,, $$
where $\omega^\alpha_{IJ} = \dparder{f_I^\alpha}{x^J} - \dparder{f_J^\alpha}{x^I}$. Consider a map $\psi\colon D\subset\R^k\to M$ with local expression $\psi(s) = (t^\alpha(s),x^I(s),z^\alpha(s))$. Then,
$$ \psi'_\alpha = \left(t^\beta, x^I, z^\beta; \parder{t^\beta}{s^\alpha}, \parder{x^I}{s^\alpha}, \parder{z^\beta}{s^\alpha}\right)\,. $$
Then, the Hamilton--De Donder--Weyl equations in adapted coordinates read
\begin{equation}
    \begin{dcases}
        \parder{x^J}{t^\alpha}\omega_{JI}^\alpha = \left(\parder{h}{x^I} + \parder{h}{s^\alpha}f_I^\alpha\right)\circ\psi\,,\\
        \parder{s^\alpha}{t^\alpha} - f_I^\alpha\parder{x^I}{t^\alpha} = -h\circ\psi\,,\\
        \parder{t^\alpha}{s^\beta} = \delta^\alpha_\beta\,.
    \end{dcases}
\end{equation}

On the other hand, if the local expression in Darboux coordinates of a map $\psi\colon D\subset\R^k\to M$ is $\psi(r) = (t^\alpha(r), q^i(r), p_i^\alpha(r), z^\alpha(r))$, where $r = (r^1,\dotsc,r^k)\in\R^k$. Then, the Hamilton--De Donder--Weyl equations in Darboux coordinates read
\begin{equation}\label{eq:HDW-map-Darboux}
    \begin{dcases}
        \parder{t^\beta}{r^\alpha} = \delta_\alpha^\beta\,,\\
        \parder{q^i}{r^\alpha} = \parder{h}{p_i^\alpha}\circ\psi\,,\\
        \parder{p_i^\alpha}{r^\alpha} = -\left( \parder{h}{q^i} + p_i^\alpha\parder{h}{z^\alpha} \right)\circ\psi\,,\\
        \parder{z^\alpha}{r^\alpha} = \left( p_i^\alpha\parder{h}{p_i^\alpha} - h \right)\circ\psi\,.
    \end{dcases}
\end{equation}

\begin{dfn}
    Consider a $k$-cocontact Hamiltonian system $(M,\tau^\alpha,\eta^\alpha, h)$. The \textsl{$k$-cocontact Hamilton--De Donder--Weyl equations for a $k$-vector field} $\bfX = (X_\alpha)\in\X^k(M)$ are
    \begin{equation}\label{eq:HDW-field}
        \begin{dcases}
            i(X_\alpha)\d\eta^\alpha = \d h - (\Lie_{R_\alpha^t}h)\tau^\alpha - (\Lie_{R_\alpha^z}h)\eta^\alpha\,,\\
            i(X_\alpha)\eta^\alpha = -h\,,\\
            i(X_\alpha)\tau^\beta = \delta_\alpha^\beta\,.
        \end{dcases}
    \end{equation}
    A $k$-vector field solution to these equations is a \textsl{$k$-cocontact Hamiltonian $k$-vector field}. We will denote this set of $k$-vector fields by $\Xh^k(M)$.
\end{dfn}

\begin{prop}\label{prop:k-cocontact-HDW-have-solutions}
    The $k$-cocontact Hamilton--De Donder--Weyl equations \eqref{eq:HDW-field} admit solutions. They are not unique if $k > 1$.
\end{prop}
\begin{proof}
    Consider the bundle maps
    $$ \rho:\T M\to\textstyle\bigoplus\nolimits^k\cT M\,,\qquad \sigma: \bigoplus\nolimits^k\T M \to \cT M\,, $$
    given by
    $$ \rho(X) = (i_X\d\eta^1,\dotsc,i_X\d\eta^k)\,,\qquad \sigma(X_1,\dotsc,X_k) = i_{X_\alpha}\d\eta^\alpha\,. $$
    These morphisms can be extended to $\Cinfty(M)$-modules. Notice that $\ker\rho = \D^\rmR$ is the Reeb distribution. Using the natural identification $(E\oplus F)^* = E^*\oplus F^*$, the transposed morphism of $\tau$ is $\transp{\tau} = -\rho$, taking into account that $\transp{\d\eta^\alpha} = -\d\eta^\alpha$.

    The first Hamilton--De Donder--Weyl equation for a $k$-vector field $\bfX$ can be written as
    $$ \tau\circ\bfX = \d h - R_\alpha^z(h)\eta^\alpha - R_\alpha^t(h)\tau^\alpha\,. $$
    A sufficient condition for this linear equation to have solutions $\bfX$ is that the right-hand-side must be in the image of $\tau$, that is, anihilated by any section of $\D^\rmR = \ker\transp{\tau}$. But since
    $$ i_R(\d h - R_\alpha^z(h)\eta^\alpha - R_\alpha^t(h)\tau^\alpha) = 0\,,\quad\text{for every }R\in\D^\rmR\,, $$
    we can conclude that the first Hamilton-De Donder--Weyl has solutions. Notice that if $\bfX$ is a solution to the first equation, $\bfX + \mathbf{R}$, where $\mathbf{R}$ is a $k$-vector field whose components are in $\D^\rmR$, is also a solution. On the other hand, the second and third equations have common solutions $\mathbf{R}$ whose components belong to the Reeb distribution, for instance $\mathbf{R} = (-h R_1^z + R_1^t, R_2^t,\dotsc,R_k^t)$.

    The non-uniqueness for $k>1$ is obvious.
\end{proof}

Consider a $k$-vector field $\bfX = (X_1,\dotsc,X_k)\in\X^k(M)$ with local expression in adapted coordinates
$$ X_\alpha = A_\alpha^\beta\parder{}{t^\beta} + B_\alpha^I\parder{}{x^I} + D_\alpha^\beta\parder{}{z^\beta}\,. $$
Thus, equations \eqref{eq:HDW-field} in adapted coordinates read
\begin{equation}
    \begin{dcases}
        A_\alpha^\beta = \delta_\alpha^\beta\,,\\
        B_\alpha^J\omega_{JI}^\alpha = \parder{h}{x^I} + \parder{h}{z^\alpha}f_I^\alpha\,,\\
        D_\alpha^\alpha - f_I^\alpha B_\alpha^I = -h\,.
    \end{dcases}
\end{equation}

On the other hand, consider a $k$-vector field $\bfX = (X_1,\dotsc,X_k)\in\X^k(M)$ with local expression in Darboux coordinates
$$ X_\alpha = A_\alpha^\beta\parder{}{t^\beta} + B_\alpha^i\parder{}{q^i} + C_{\alpha i}^\beta\parder{}{p_i^\beta} + D_\alpha^\beta\parder{}{z^\beta}\,. $$
Imposing equations \eqref{eq:HDW-field}, we get the conditions
\begin{equation}\label{eq:HDW-field-Darboux}
    \begin{dcases}
        A_\alpha^\beta = \delta_\alpha^\beta\,,\\
        B_\alpha^i = \parder{h}{p_i^\alpha}\,,\\
        C_{\alpha i}^\alpha = -\left( \parder{h}{q^i} + p_i^\alpha\parder{h}{z^\alpha} \right)\,,\\
        D_\alpha^\alpha = p_i^\alpha \parder{h}{p_i^\alpha} - h\,.
    \end{dcases}
\end{equation}

\begin{prop}
    Let $\bfX \in\X^k(M)$ be an integrable $k$-vector field. Then $\bfX$ is a solution to \eqref{eq:HDW-field} if and only if every integral section of $\bfX$ satisfies the $k$-cocontact Hamilton--De Donder--Weyl equations \eqref{eq:HDW-map}.
\end{prop}
\begin{proof}
    Recall that since $\bfX$ is integrable, every point of $M$ is in the image of an integral section of $\bfX$. The proposition is a direct consequence of this fact and of equations \eqref{eq:HDW-map} and \eqref{eq:HDW-field}.
\end{proof}

It is worth noting that, as in the $k$-symplectic and $k$-contact cases, equations \eqref{eq:HDW-map} and \eqref{eq:HDW-field} are not completely equivalent since a solution to \eqref{eq:HDW-map} may not be an integral section of an integrable $k$-vector field $\bfX$ solution to equations \eqref{eq:HDW-field}.

The following proposition provides an alternative way of writing the $k$-cocontact Hamilton--De Donder--Weyl equations for $k$-vector fields.

\begin{prop}
    The $k$-cocontact Hamilton--De Donder--Weyl equations \eqref{eq:HDW-field} are equivalent to
    \begin{equation}
        \begin{dcases}
            \Lie_{X_\alpha}\eta^\alpha = -(\Lie_{R_\alpha^t}h)\tau^\alpha - (\Lie_{R_\alpha^z}h)\eta^\alpha\,,\\
            i(X_\alpha)\eta^\alpha = -h\,,\\
            i(X_\alpha)\tau^\beta = \delta_\alpha^\beta\,.
        \end{dcases}
    \end{equation}
\end{prop}

\section{Lagrangian formalism}\label{sec:5}

In this section we devise the Lagrangian counterpart of the formulations introduced in the previous section. We begin by introducing the geometric structures of the phase bundle and defining the notion of second-order partial differential equation. In second place, we develop the Lagrangian formalism and introduce the $k$-cocontact Euler--Lagrange equations as the Hamilton--De Donder--Weyl of a $k$-cocontact Lagrangian system.

\subsection{Geometry of the phase bundle}

The phase space for the Lagrangian counterpart of the $k$-cocontact formalism will be the product bundle $M = \R^k\times\bigoplus^k\T Q\times\R^k$ endowed with natural coordinates $(t^\alpha, q^i, v^i_\alpha, z^\alpha)$. We have the natural projections
\begin{align*}
    \tau_1^\alpha&\colon M\to\R\ ,&& \tau_1^\alpha(t^1,\dotsc,t^k, {v_q}_1, \dotsc, {v_q}_k, z^1,\dotsc,z^k) = t^\alpha\,,\\
    \tau_2&\colon M\to\textstyle\bigoplus\nolimits^k\T Q\ ,&& \tau_2(t^1,\dotsc,t^k, {v_q}_1, \dotsc, {v_q}_k, z^1,\dotsc,z^k) = ({v_q}_1, \dotsc, {v_q}_k)\,,\\
    \tau_2^\alpha&\colon M\to\T Q\ ,&& \tau_2^\alpha(t^1,\dotsc,t^k, {v_q}_1, \dotsc, {v_q}_k, z^1,\dotsc,z^k) = {v_q}_\alpha\,,\\
    \tau^\alpha&\colon \textstyle\bigoplus\nolimits^k\T Q\to\T Q\ ,&& \tau^\alpha(t^1,\dotsc,t^k, {v_q}_1, \dotsc, {v_q}_k, z^1,\dotsc,z^k) = {v_q}_\alpha\,,\\
    \tau_3^\alpha&\colon M\to\R\ ,&& \tau_3^\alpha(t^1,\dotsc,t^k, {v_q}_1, \dotsc, {v_q}_k, z^1,\dotsc,z^k) = z^\alpha\,,\\
    \tau_0&\colon M\to \R^k\times Q\times\R^k\ ,&& \tau_0(t^1,\dotsc,t^k, {v_q}_1, \dotsc, {v_q}_k, z^1,\dotsc,z^k) = (t^1,\dotsc,t^k, q, z^1,\dotsc,z^k)\,, 
\end{align*}
which can be summarized in the following diagram:
\begin{center}
    \begin{tikzcd}
        \R & \R^k\times\bigoplus^k\T Q\times\R^k \arrow[l, swap, "\tau_1^\alpha", ] \arrow[r, "\tau_3^\alpha"] \arrow[dd, "\tau_2"] \arrow[ddd, bend right=50, swap, "\tau_0"] \arrow[ddr, "\tau_2^\alpha"] & \R \\
        &  && \\
        & \bigoplus^k\T Q \arrow[r, swap, "\tau^\alpha"] & \T Q \\
        & \R^k\times Q\times\R^k 
    \end{tikzcd}
\end{center}

Since the bundle $\tau_2 \colon \R^k\times\bigoplus^k\T Q\times\R^k\to\bigoplus^k\T Q$ is trivial, the canonical structures in $\bigoplus^k\T Q$, namely the canonical $k$-tangent structure $(J^\alpha)$ and the Liouville vector field $\Delta$, can be extended to $\R^k\times\bigoplus^k\T Q\times\R^k$ in a natural way. Their local expression remain the same:
$$ J^\alpha = \parder{}{v^i_\alpha}\otimes\d q^i\,,\qquad \Delta = v^i_\alpha\parder{}{v^i_\alpha}\,. $$
These canonical structures can be used to extend the notion of \textsc{sopde} (second-order partial differential equation) to the bundle $\R^k\times\bigoplus^k\T Q\times\R^k$:

\begin{dfn}
    A $k$-vector field $\mathbf{\Gamma} = (\Gamma_\alpha)\in\X^k(\R^k\times\bigoplus^k\T Q\times\R^k)$ is a \textsl{second-order partial differential equation} or \textsc{sopde} if $J^\alpha(\Gamma_\alpha) = \Delta$.
\end{dfn}

A straightforward computations shows that the local expression of a \textsc{sopde} reads
$$ \Gamma_\alpha = A_\alpha^\beta\parder{}{t^\beta} + v^i_\alpha\parder{}{q^i} + C_{\alpha\beta}^i\parder{}{v_\beta^i} + D_\alpha^\beta\parder{}{z^\beta} \,. $$

\begin{dfn}
    Consider a map $\psi\colon\R^k\to\R^k\times Q\times\R^k$ with $\psi = (t^\alpha, \phi, z^\alpha)$, where $\phi\colon\R^k\to Q$. The \textsl{first prolongation} of $\psi$ to $\R^k\times\bigoplus^k\T Q\times\R^k$ is the map $\psi'\colon\R^k\to\R^k\times\bigoplus^k\T Q\times\R^k$ given by $\psi' = (t^\alpha,\phi',z^\alpha)$, where $\phi'$ is the first prolongation of $\phi$ to $\bigoplus^k\T Q$. The map $\psi'$ is said to be \textsl{holonomic}.
\end{dfn}
Let $\psi\colon\R^k\to\R^k\times Q\times\R^k$ be a map with local expression $\psi(r) = (t^\alpha(r), q^i(r), z^\alpha(r))$, where $r\in\R^k$. Then, its first prolongation has local expression
$$ \psi'(r) = \left( t^\alpha(r), q^i(r), \parder{q^i}{r^\alpha}(r), z^\alpha(r) \right)\,. $$

\begin{prop}
    An integrable $k$-vector field $\mathbf{\Gamma} \in\X^k(\R^k\times\bigoplus^k\T Q\times\R^k)$ is a \textsc{sopde} if and only if its integral sections are holonomic.
\end{prop}

It is important to point out that the product manifold $\R^k\times\bigoplus^k\T Q\times\R^k$ does not have a canonical $k$-cocontact structure, in contrast to what happens to the manifold $\R^k\times\bigoplus^k\cT Q\times\R^k$, where we do have a natural $k$-cocontact structure as seen in Example \ref{ex:canonical -k-cocontact-structure}. In what follows we will show that, in favourable cases, given a Lagrangian function $L$ defined on $\R^k\times\bigoplus^k\T Q\times\R^k$ one can build up a $k$-cocontact structure.

\begin{dfn}
    A \textsl{Lagrangian function} on $\R^k\times\bigoplus^k\T Q\times\R^k$ is a function $L\colon\R^k\times\bigoplus^k\T Q\times\R^k\to\R$.
    \begin{itemize}
        \item The \textsl{Lagrangian energy} associated to the Lagrangian function $L$ is the function $E_L\in\Cinfty(\R^k\times\bigoplus^k\T Q\times\R^k)$ given by $E_L = \Delta(L) - L$.
        \item The \textsl{Cartan forms} associated to the Lagrangian $L$ are
        $$ \theta_L^\alpha = \transp{J^\alpha}\circ\d L\in\Omega^1(\R^k\times\textstyle\bigoplus\nolimits^k\T Q\times\R^k)\,,\qquad \omega_L^\alpha = -\d\theta_L^\alpha\in\Omega^2(\R^k\times\textstyle\bigoplus\nolimits^k\T Q\times\R^k)\,, $$
        where $\transp{J^\alpha}$ denotes the transpose of $J^\alpha$.
        \item The \textsl{contact forms} associated to the Lagrangian $L$ are
        $$ \eta_L^\alpha = \d z^\alpha - \theta_L^\alpha\in\Omega^1(\R^k\times\textstyle\bigoplus\nolimits^k\T Q\times\R^k)\,. $$
        \item The couple $(\R^k\times\bigoplus^k\T Q\times\R^k, L)$ is a $k$-\textsl{cocontact Lagrangian system}.
    \end{itemize}
\end{dfn}

It is clear that $\d\eta_L^\alpha = \omega_L^\alpha$. The local expressions in natural coordinates $(t^\alpha, q^i, v_\alpha^i, z^\alpha)$ of the objects introduced in the previous definition are
\begin{align}
    E_L &= v_\alpha^i\parder{L}{v_\alpha^i} - L\,,\\
    \theta^\alpha_L &= \parder{L}{v_\alpha^i}\d q^i\,,\\
    \eta_L^\alpha &= \d z^\alpha - \parder{L}{v_\alpha^i}\d q^i\,,\\
    \d\eta_L^\alpha &= \parderr{L}{t^\beta}{v_\alpha^i}\d q^i\wedge\d t^\beta + \parderr{L}{q^j}{v_\alpha^i}\d q^i\wedge\d q^j + \parderr{L}{v_\beta^j}{v_\alpha^i}\d q^i\wedge\d v_\beta^j + \parderr{L}{z^\beta}{v_\alpha^i}\d q^i\wedge\d z^\beta\,.
\end{align}

Before introducing the Legendre map associated to a Lagrangian function, let us recall the notion of fibre derivative. Given two vector bundles $E,F$ over the same base manifold $B$ and a bundle map $f\colon E\to F$, the \textsl{fibre derivative} of $f$ is the map $\F f\colon E\longrightarrow\Hom(E,F)\cong F\otimes E^\ast$ obtained by restricting the map $f$ to the fibers $f_b\colon E_b\to F_b$ and computing the usual derivative: $\F f(e_b) = D f_b(e_b)$. If the second vector bundle is trivial and has rank 1, namely for a function $f\colon E\to\R$, then $\F f\colon E\to E^\ast$. This fibre derivative has a fibre derivative $\F(\F f) = \F^2 f\colon E\to E^\ast\otimes E^\ast$, called the \textsl{fibre Hessian} of $f$. For every $e_b\in E_b\subset E$, $\F^2 f(e_b)$ is a symmetric bilinear form on $E_b$. The fibre derivative $\F f$ is a local diffeomorphism at a point $e\in E$ if and only if the Hessian $\F^2 f(e)$ is non-degenerate (see \cite{Gra2000} for more details).

\begin{dfn}
    Given a Lagrangian function $L\colon\R^k\times\bigoplus^k\T Q\times\R^k\to\R$, the \textsl{Legendre map} of $L$ is its fibre derivative as a function on the vector bundle $\tau_0\colon\R^k\times\bigoplus^k\T Q\times\R^k\to \R^k\times Q\times\R^k$. Namely, the Legendre map of a Lagrangian function $L\colon\R^k\times\bigoplus^k\T Q\times\R^k\to\R$ is the map
    $$ \F L\colon \R^k\times\textstyle\bigoplus\nolimits^k\T Q\times\R^k\longrightarrow\R^k\times\textstyle\bigoplus\nolimits^k\cT Q\times\R^k $$
    given by
    $$ \F L(t, {v_q}_1,\dotsc, {v_q}_k, z) = (t, \F L(t, \cdot, z)({v_q}_1,\dotsc, {v_q}_k), z)\,, $$
    where $\F L(t, \cdot, z)$ denotes the Lagrangian function with $t$ and $z$ freezed.
\end{dfn}

In natural coordinates $(t^\alpha, q^i, v_\alpha^i,z^\alpha)$, the Legendre map has local expression
$$
    \F L(t^\alpha, q^i, v^i_\alpha, z^\alpha) = \left( t^\alpha, q^i, \parder{L}{v^i_\alpha}, z^\alpha \right)\,.
$$

\begin{prop}
    The Cartan forms satisfy
    $$ \theta_L^\alpha = (\pi_2^\alpha\circ\F L)^\ast\theta\,,\qquad \omega_L^\alpha = (\pi_2^\alpha\circ\F L)^\ast\omega\,, $$
    where $\theta\in\Omega^1(\cT Q)$ and $\omega = -\d\theta\in\Omega^2(\cT Q)$ are the Liouville and symplectic canonical forms of the cotangent bundle $\cT Q$.
\end{prop}

The regularity of the Legendre map characterizes the Lagrangian functions which yield $k$-cocontact structures on the phase bundle $\R^k\times\bigoplus^k\T Q\times\R^k$.

\begin{prop}\label{prop:regular-lagrangian}
    Consider a Lagrangian function $L\colon \R^k\times\bigoplus^k\T Q\times\R^k\to\R$. The following are equivalent:
    \begin{enumerate}[{\rm(1)}]
        \item The Legendre map $\F L$ is a local diffeomorphism.
        \item The fibre Hessian of the Lagrangian $L$, namely the map
        $$ \F^2 L\colon \R^k\times\textstyle\bigoplus\nolimits^k\T Q\times\R^k\longrightarrow(\R^k\times\bigoplus\nolimits^k\cT Q\times\R^k)\otimes(\R^k\times\textstyle\bigoplus\nolimits^k\cT Q\times\R^k)\,, $$
        is everywhere nondegenerate, where the tensor product is of vector bundles over $\R^k\times Q\times\R^k$.
        \item The family $(\tau^\alpha = \d t^\alpha, \eta^\alpha_L)$ is a $k$-cocontact structure on $\R^k\times\bigoplus^k\T Q\times\R^k$.
    \end{enumerate}
\end{prop}
\begin{proof}
    Taking natural coordinates $(t^\alpha, q^i, v^i_\alpha, z^\alpha)$, We have
    \begin{align}
        \F^2 L(t^\alpha, q^i, v^i_\alpha, z^\alpha) &= \left( t^\alpha, q^i, W_{ij}^{\alpha\beta}, z^\alpha \right)\,,\quad \text{where } W_{ij}^{\alpha\beta} = \left( \parderr{L}{v^i_\alpha}{v^j_\beta} \right)\,.
    \end{align}
    The conditions in the proposition mean that the matrix $W = (W_{ij}^{\alpha\beta})$ is everywhere nonsingular.
\end{proof}

\begin{dfn}
    A Lagrangian function $L\colon \R^k\times\bigoplus^k\T Q\times\R^k\to\R$ is said to be \textsl{regular} if the equivalent statements in Proposition \ref{prop:regular-lagrangian} hold. Otherwise $L$ is said to be \textsl{singular}. In addition, if the Legendre map $\F L$ is a global diffeomorphism, $L$ is a \textsl{hyperregular} Lagrangian.
\end{dfn}

Let $(\R^k\times\bigoplus^k\T Q\times\R^k, L)$ be a regular $k$-cocontact Lagrangian system. By Theorem \ref{thm:Reeb-vector-fields}, the Reeb vector fields $(R^t_L)_\alpha, (R^z_L)_\alpha\in\X(\R^k\times\bigoplus^k\T Q\times\R^k)$ are uniquely given by the relations
\begin{gather}
    i\left((R^t_L)_\alpha\right)\d\eta_L^\beta = 0\,,\qquad i\left((R^t_L)_\alpha\right)\eta_L^\beta = 0\,,\qquad i\left((R^t_L)_\alpha\right)\d t^\beta = \delta_\alpha^\beta\,,\\
    i\left((R^z_L)_\alpha\right)\d\eta_L^\beta = 0\,,\qquad i\left((R^z_L)_\alpha\right)\eta_L^\beta = \delta_\alpha^\beta\,,\qquad i\left((R^z_L)_\alpha\right)\d t^\beta = 0\,.
\end{gather}
The local expressions of the Reeb vector fields are
\begin{gather}
    (R^t_L)_\alpha = \parder{}{t^\alpha} - W_{\gamma\beta}^{ji}\parderr{L}{t^\alpha}{v^j_\gamma}\parder{}{v^i_\beta}\,,\\
    (R^z_L)_\alpha = \parder{}{z^\alpha} - W_{\gamma\beta}^{ji}\parderr{L}{z^\alpha}{v^j_\gamma}\parder{}{v^i_\beta}\,,
\end{gather}
where $W_{\alpha\beta}^{ij}$ is inverse of the Hessian matrix $W_{ij}^{\alpha\beta} = \left( \dparderr{L}{v^i_\alpha}{v^j_\beta} \right)$, namely
$$ W_{\alpha\beta}^{ij}\parderr{L}{v^j_\beta}{v^k_\gamma} = \delta^i_k\delta^\gamma_\alpha\,. $$

\subsection{\texorpdfstring{$k$}--cocontact Euler--Lagrange equations}

We have proved in the previous section that every regular $k$-cocontact Lagrangian system $(\R^k\times\bigoplus^k\T Q\times\R^k,L)$ yields the $k$-cocontact Hamiltonian system $(\R^k\times\bigoplus^k\T Q\times\R^k,\tau^\alpha = \d t^\alpha, \eta^\alpha, E_L)$. Taking this into account, we can define:

\begin{dfn}
    Let $(\R^k\times\bigoplus^k\T Q\times\R^k,L)$ be a $k$-cocontact Lagrangian system. The \textsl{$k$-cocontact Euler--Lagrange equations for a holonomic map} $\psi\colon\R^k\to \R^k\times\bigoplus^k\T Q\times\R^k$ are
    \begin{equation}\label{eq:EL-map}
        \begin{dcases}
            i(\psi_\alpha')\d\eta_L^\alpha = \left(\d E_L - (\Lie_{(R^t_L)_\alpha}E_L)\d t^\alpha - (\Lie_{(R^z_L)_\alpha}E_L)\eta_L^\alpha\right)\circ\psi\,,\\
            i(\psi_\alpha')\eta_L^\alpha = -E_L\circ\psi\,,\\
            i(\psi_\alpha')\d t^\beta = \delta_\alpha^\beta\,.
        \end{dcases}
    \end{equation}
    The \textsl{$k$-cocontact Lagrangian equations for a $k$-vector field} $\ \bfX = (X_\alpha)\in\X^k(\R^k\times\bigoplus^k\T Q\times\R^k)$ are
    \begin{equation}\label{eq:EL-field}
        \begin{dcases}
            i(X_\alpha)\d\eta_L^\alpha = \d E_L - (\Lie_{(R^t_L)_\alpha}E_L)\d t^\alpha - (\Lie_{(R^z_L)_\alpha}E_L)\eta_L^\alpha\,,\\
            i(X_\alpha)\eta_L^\alpha = -E_L\,,\\
            i(X_\alpha)\d t^\beta = \delta_\alpha^\beta\,.
        \end{dcases}
    \end{equation}
    A $k$-vector field $\bfX$ solution to equations \eqref{eq:EL-field} is said to be a \textsl{$k$-contact Lagrangian vector field}.
\end{dfn}

The next proposition states that, if the Lagrangian $L$ is regular, the Lagrangian equations \eqref{eq:EL-field} always have solutions, although they are not unique in general. It is a direct translation of Proposition \ref{prop:k-contact-HDW-have-solutions} to the Lagrangian language.

\begin{prop}
    Consider a regular $k$-cocontact Lagrangian system $(\R^k\times\bigoplus^k\T Q\times\R^k, L)$. Then, the $k$-cocontact Lagrangian equations \eqref{eq:EL-field} admit solutions. They are not unique if $k>1$.
\end{prop}

Consider a map $\psi\colon\R^k\to\R^k\times\bigoplus^k\T Q\times\R^k$ with local expression in natural coordinates $\psi(r) = (t^\alpha(r),q^i(r),v^i_\alpha(r),z^\alpha(r))$, where $r = (r^1,\dotsc,r^k)\in\R^k$. Then, equations \eqref{eq:EL-map} for the map $\psi$ read
\begin{equation}\label{eq:EL-map-coordinates}
    \begin{dcases}
        \parder{t^\beta}{r^\alpha} = \delta_\alpha^\beta\,,\\
        \parder{}{r^\alpha}\left( \parder{L}{v^i_\alpha} \circ\psi\right) = \left( \parder{L}{q^i} + \parder{L}{z^\alpha}\parder{L}{v^i_\alpha} \right)\circ\psi\,,\\
        \parder{(z^\alpha)}{r^\alpha} = L\circ\psi\,.
    \end{dcases}
\end{equation}
For a $k$-vector field $\bfX = (X_\alpha)\in\X^k(\R^k\times\bigoplus^k\T Q\times\R^k)$, with local expression in natural coordinates
$$ X_\alpha = A_\alpha^\beta\parder{}{t^\beta} + B_\alpha^i\parder{}{q^i} + C_{\alpha \beta}^i\parder{}{v_\beta^i} + D_\alpha^\beta\parder{}{z^\beta}\,, $$
equations \eqref{eq:EL-field} read
\begin{align}
    0 &= A_\alpha^\beta - \delta_\alpha^\beta\,,\label{eq:k-contact-Lagrangian-1}\\
    0 &= \left( B_\alpha^j - v_\alpha^j \right)\parderr{L}{v^j_\alpha}{z^\beta}\,,\label{eq:k-contact-Lagrangian-2}\\
    0 &= \left( B_\alpha^j - v_\alpha^j \right)\parderr{L}{v^j_\alpha}{t^\beta}\,,\label{eq:k-contact-Lagrangian-2.5}\\
    0 &= \left( B_\alpha^j - v_\alpha^j \right)\parderr{L}{v^i_\beta}{v^j_\alpha}\,,\label{eq:k-contact-Lagrangian-3}\\
    0 &= \left( B_\alpha^j - v_\alpha^j \right)\parderr{L}{q^i}{v^j_\alpha} + \parder{L}{q^i} - \parderr{L}{t^\alpha}{v_\alpha^i} - \parderr{L}{q^j}{v^i_\alpha}B_\alpha^j \nonumber\\
    & \qquad\qquad\qquad\qquad\qquad\qquad - \parderr{L}{v^j_\beta}{v^i_\alpha}C_{\alpha\beta}^j - \parderr{L}{z^\beta}{v^i_\alpha}D_\alpha^\beta + \parder{L}{z^\alpha}\parder{L}{v^i_\alpha} \,,\label{eq:k-contact-Lagrangian-4}\\
    0 &= L + \parder{L}{v^i_\alpha}\left( B_\alpha^i - v^i_\alpha \right) - D_\alpha^\alpha\,.\label{eq:k-contact-Lagrangian-5}
\end{align}
If the Lagrangian function $L$ is regular, equations \eqref{eq:k-contact-Lagrangian-3} yield the conditions $B_\alpha^i = v_\alpha^i$, namely the $k$-vector field $\bfX$ has to be a \textsc{sopde}. In this case, equations \eqref{eq:k-contact-Lagrangian-2} and \eqref{eq:k-contact-Lagrangian-2.5} hold identically and equations \eqref{eq:k-contact-Lagrangian-1}, \eqref{eq:k-contact-Lagrangian-4} and \eqref{eq:k-contact-Lagrangian-5} yield
\begin{align}
    A_\alpha^\beta &= \delta_\alpha^\beta\,,\label{eq:k-contact-Lagrangian-regular-1}\\
    \parder{L}{q^i} + \parder{L}{z^\alpha}\parder{L}{v^i_\alpha} &= \parderr{L}{t^\alpha}{v_\alpha^i} + \parderr{L}{q^j}{v^i_\alpha}v^j_\alpha + \parderr{L}{v^j_\beta}{v^i_\alpha}C_{\alpha\beta}^j + \parderr{L}{z^\beta}{v^i_\alpha}D_\alpha^\beta\,,\label{eq:k-contact-Lagrangian-regular-2}\\
	D_\alpha^\alpha &= \L\,.\label{eq:k-contact-Lagrangian-regular-3}
\end{align}
If the \textsc{sopde} $\bfX$ is integrable, equations \eqref{eq:k-contact-Lagrangian-regular-1}, \eqref{eq:k-contact-Lagrangian-regular-2} and \eqref{eq:k-contact-Lagrangian-regular-3} are the Euler--Lagrange equations \eqref{eq:EL-map-coordinates} for its integral maps. Thus, we have proved the following:

\begin{prop}\label{prop:Euler-Lagrange}
    Let $L\colon \R^k\times\bigoplus^k\T Q\times\R^k\to\R$ be a regular Lagrangian and consider a Lagrangian $k$-vector field $\bfX$, namely a solution to equations \eqref{eq:EL-field}. Then $\bfX$ is a \textsc{sopde} and if, in addition, $\bfX$ is integrable, its integral sections are solutions to the $k$-cocontact Euler--Lagrange equations \eqref{eq:EL-map}.

    The \textsc{sopde} $\bfX$ is called an \textsl{Euler--Lagrange $k$-vector field} associated to the Lagrangian function $L$.
\end{prop}

\begin{rmrk}
    If the Lagrangian function $L$ is regular or hyperregular, the Legendre map $\F L$ is a (local) diffeomorphism between $\R^k\times\bigoplus^k\T Q\times\R^k$ and $\R^k\times\bigoplus^k\cT Q\times\R^k$ such that $\F L^\ast\eta^\alpha = \eta_L^\alpha$. In addition, there exists, at least locally, a function $h\in\Cinfty(\R^k\times\bigoplus^k\cT Q\times\R^k)$ such that $h \circ \F L = E_L$. Then, we have the $k$-cocontact Hamiltonian system $(\R^k\times\bigoplus^k\cT Q\times\R^k, \eta^\alpha, h)$, for which $\F L_\ast (R^t_L)_\alpha = R^t_\alpha$ and $\F L_\ast (R^z_L)_\alpha = R^z_\alpha$. If $\mathbf{\Gamma}$ is an Euler--Lagrange $k$-vector field associated to the Lagrangian function $L$ in $\R^k\times\bigoplus^k\T Q\times\R^k$, we have that the $k$-vector field $\bfX = \F L_\ast\mathbf{\Gamma}$ is a $k$-cocontact Hamiltonian $k$-vector field associated to $h$ in $\R^k\times\bigoplus^k\T Q\times\R^k$, and conversely.
\end{rmrk}

\begin{rmrk}
    In the case $k = 1$, we recover the cocontact Lagrangian formalism presented in the recent paper \cite{DeLeo2022} for time-dependent contact Lagrangian systems.
\end{rmrk}

\begin{rmrk}
    It is important to point out that the field equations obtained in this work from both the Hamiltonian and Lagrangian formalism coincide with the ones obtained by means of the so-called {\sl multicontact formalism} introduced in \cite{LGMRR-2022} as a generalization of the multisymplectic setting.
\end{rmrk}

\subsection{Lagrangian functions with holonomic damping term}

In this section, a particular type of Lagrangian functions is studied in full detail: the so-called Lagrangians with holonomic damping term \cite{Gas2019}. This family of Lagrangians is particularly interesting since it appears in many physical examples.

\begin{dfn}
    A \textsl{Lagrangian function with holonomic damping term} in $\R^k\times\bigoplus^k\T Q\times\R^k$ is a function $\L = L + \phi\in\Cinfty(\R^k\times\bigoplus^k\T Q\times\R^k)$, where $L = \bar\tau_2^\ast L_\circ$, where $\bar\tau_2\colon \R^k\times\bigoplus^k\T Q\times\R^k\to \R^k\times\bigoplus^k\T Q$ for some Lagrangian function $L_\circ\in\Cinfty(\R^k\times\bigoplus^k\T Q)$ and $\phi = \tau_0^\ast\phi_\circ$, for $\phi_\circ\in\Cinfty(\R^k\times Q\times\R^k)$.
\end{dfn}

Taking natural coordinates $(t^\alpha, q^i, v_\alpha^i, z^\alpha)$ in $\R^k\times\bigoplus^k\T Q\times\R^k$, a Lagrangian with holonomic damping term has the expression
\begin{equation}\label{eq:Lagrangian-holonomic-damping-term}
    \L(t^\alpha, q^i, v_\alpha^i, z^\alpha) = L(t^\alpha, q^i, v^i_\alpha) + \phi(t^\alpha, q^i, z^\alpha)\,.
\end{equation}

It is clear that the momenta $p_i^\alpha = \tparder{\L}{v^i_\alpha}$ defined by the Legendre map are independent of the coordinates $z^\alpha$, namely one has that $\dparderr{\L}{z^\alpha}{v^i_\beta} = 0$ for Lagrangian functions with holonomic damping term.

\begin{prop}
    Consider the Lagrangian function with holonomic damping term $\L = L + \phi$. Then, its Cartan forms, contact forms, Lagrangian energy and Reeb vector fields read
    $$ \theta_\L^\alpha = \theta_L^\alpha\,,\quad \eta_\L^\alpha = \d z^\alpha - \theta_L^\alpha\,,\quad E_\L = E_L - \phi\,,\quad (R_\L^t)_\alpha = \parder{}{t^\alpha}\,,\quad (R_\L^z)_\alpha = \parder{}{z^\alpha}\,.  $$
    where $\theta_L^\alpha$ are the Cartan one-forms of $L$ considered (via pull-back) as one-forms on $\R^k\times\bigoplus^k\T Q\times\R^k$, and $E_L$ is the energy of $L$ as a function on $\R^k\times\bigoplus^k\T Q\times\R^k$.

    The Legendre map of $\L$, namely $\F\L\colon\R^k\times\bigoplus^k\T Q\times\R^k\to\R^k\times\bigoplus^k\cT Q\times\R^k$, can be expressed as $\F\L = \F L\times\Id_\R^k$, where $\F L$ is the Legendre map of $L$. The fibred Hessians are related by $\F^2\L(t^\alpha, {v_q}_\alpha,z^\alpha) = \F^2L(t^\alpha, {v_q}_\alpha)$. Moreover, $\L$ is regular if, and only if, $L$ is regular.
\end{prop}

The proof of this proposition is straightforward by taking local coordinates. It is also clear that $\L$ is hyperregular if and only if $L$ is hyperregular. In this case, the Legendre map $\F\L$ is a diffeomorphism and one can state the canonical Hamiltonian formulation for the Lagrangian with holonomic damping term $\L = L + \phi$ via the Legendre map.

Consider the $k$-cocontact Lagrangian system $(\R^k\times\bigoplus^k\T Q\times\R^k,\L)$, where $\L = L + \phi$ is a Lagrangian function with holonomic damping term as in \eqref{eq:Lagrangian-holonomic-damping-term}. Recall that the dynamical equations for $k$-vector fields of this system are
\begin{equation}
    \begin{dcases}
        i(X_\alpha)\d\eta_\L^\alpha = \d E_\L - (\Lie_{(R^t_\L)_\alpha}E_\L)\d t^\alpha - (\Lie_{(R^z_\L)_\alpha}E_\L)\eta_\L^\alpha\,,\\
        i(X_\alpha)\eta_\L^\alpha = -E_\L\,,\\
        i(X_\alpha)\d t^\beta = \delta_\alpha^\beta\,.
    \end{dcases}
\end{equation}
Take natural coordinates $(t^\alpha, q^i, v_\alpha^i, z^\alpha)$ in $\R^k\times\bigoplus^k\T Q\times\R^k$ and consider a $k$-vector field $\bfX = (X_\alpha)\in\X^k(\R^k\times\bigoplus^k\T Q\times\R^k)$ with local expression
$$ X_\alpha = A_\alpha^\beta\parder{}{t^\beta} + B_\alpha^i\parder{}{q^i} + C_{\alpha\beta}^i\parder{}{v_\beta^i} + D_\alpha^\beta\parder{}{z^\beta}\,. $$
Then, the second and third Lagrangian equations for the $k$-vector field $\bfX$ read
$$ A_\alpha^\beta = \delta_\alpha^\beta\,,\qquad 0 = \L + \parder{L}{v^i_\alpha}\left( B_\alpha^i - v^i_\alpha \right) - D_\alpha^\alpha\,, $$
and this is equation \eqref{eq:k-contact-Lagrangian-5} for the Lagrangian function $\L = L + \phi$. The first Lagrangian equation for $k$-vector fields yields
\begin{align}
    \left( B_\alpha^j - v_\alpha^j \right)\parderr{L}{v^i_\beta}{v^j_\alpha} &= 0\,,\label{eq:EL-holonomic-term-1}\\
    \left( \parderr{L}{q^i}{v^j_\alpha} - \parderr{L}{q^j}{v^i_\alpha} \right)B_\alpha^j - \parderr{L}{t^\alpha}{v_\alpha^i} - \parderr{L}{q^i}{v^j_\alpha} v^j_\alpha  - \parderr{L}{v^j_\beta}{v^i_\alpha}C_{\alpha\beta}^j &= - \parder{L}{q^i} - \parder{\phi}{q^i} - \parder{\phi}{z^\alpha}\parder{L}{v^i_\alpha}\,,\label{eq:EL-holonomic-term-2}
\end{align}
which correspond to equation \eqref{eq:k-contact-Lagrangian-4} for the Lagrangian $\L$. Notice that equations \eqref{eq:k-contact-Lagrangian-2} are identities since $\dparderr{L}{v^j_\alpha}{z^\beta} = 0$.

Finally, as in Proposition \ref{prop:Euler-Lagrange}, if the Lagrangian function $\L$ is regular, namely if $L$ is regular, equation \eqref{eq:EL-holonomic-term-1} implies that $B_\alpha^j = v_\alpha^j$. Thus, the $k$-vector field is a \textsc{sopde} and the dynamical equations become
\begin{align}
    \parder{t^\alpha}{r^\beta} &= \delta_\alpha^\beta\,,\\
    \parder{z^\alpha}{r^\alpha} &= \L\,,\\
    \parderr{L}{v^j_\beta}{v^i_\alpha}\parderr{q^j}{r^\alpha}{r^\beta} + \parderr{L}{q^j}{v^i_\alpha}\parder{q^j}{r^\alpha} + \parderr{L}{t^\alpha}{v^i_\alpha} - \parder{L}{q^i} = \parder{}{r^\alpha}\left( \parder{L}{v^i_\alpha} \right) - \parder{L}{q^i} &= \parder{\phi}{q^i} + \parder{\phi}{z^\alpha}\parder{L}{v^i_\alpha}\,.
\end{align}
These are the expression in natural coordinates of the Euler--Lagrange equations \eqref{eq:EL-map-coordinates} for the Lagrangian with holonomic damping term $\L = L + \phi$.

\section{\texorpdfstring{$k$}--contact systems versus autonomous \texorpdfstring{$k$}--cocontact systems}\label{sec:6}

In this section we are going to compare the $k$-contact and $k$-cocontact formulations of field theories. We will work with the canonical manifolds $\bigoplus^k\cT Q\times\R^k$ and $\R^k\times\bigoplus^k\cT Q\times\R^k$. However, due to the Darboux theorems, the results can easily be extended to the case $M$ and $\R^k\times M$ being $M$ a general $k$-contact manifold. These two canonical manifolds are related by the canonical projection $\bar\pi_2\colon \R^k\times\bigoplus^k\cT Q\times\R^k\to \bigoplus^k\cT Q\times\R^k$. We will denote by $\bar\eta^\alpha$ and $\eta^\alpha$ the canonical contact one-forms of $\R^k\times\bigoplus^k\cT Q\times\R^k$ and $\bigoplus^k\cT Q\times\R^k$ respectively. They are related by the relations $\bar\eta^\alpha = \bar\pi_2^\ast\eta^\alpha$ and have the same local expression $\eta^\alpha = \d z^\alpha - p_i^\alpha\d q^i$. The Reeb vector fields will be denoted by $\bar R_\alpha^z$ and $R_\alpha^z$ and have local expression $\tparder{}{z^\alpha}$.

\begin{dfn}
    A $k$-cocontact Hamiltonian system $(\R^k\times\bigoplus^k\cT Q\times\R^k,\d t^\alpha, \eta^\alpha, h)$ is said to be \textsl{autonomous} if $R_\alpha^t(h) = \tparder{h}{t^\alpha} = 0$ for every $\alpha = 1,\dotsc,k$.
\end{dfn}
Notice that if a Hamiltonian function $h$ does not depend on the variables $t^\alpha$, there exists a function $h_\circ\in\Cinfty(\bigoplus^k\cT Q\times\R^k)$ such that $h = \bar\pi_2^\ast h_\circ$.

For an autonomous $k$-cocontact Hamiltonian system, equations \eqref{eq:HDW-field} read
\begin{equation}\label{eq:HDW-field-autonomous}
    \begin{dcases}
        i(X_\alpha)\d\eta^\alpha = \d h - (\Lie_{R_\alpha^z}h)\eta^\alpha\,,\\
        i(X_\alpha)\eta^\alpha = -h\,,\\
        i(X_\alpha)\tau^\beta = \delta_\alpha^\beta\,.
    \end{dcases}
\end{equation}

\begin{prop}
    Every autonomous $k$-cocontact Hamiltonian system $(\R^k\times\bigoplus^k\cT Q\times\R^k, h)$ defines a $k$-contact Hamiltonian system $(\bigoplus^k\cT Q\times\R^k, h_\circ)$, where $h = \bar\pi_2^\ast H_\circ$, and conversely.
\end{prop}

\begin{thm}\label{thm:equiv-autonomous-maps}
    Consider an autonomous $k$-cocontact Hamiltonian system $(\R^k\times\bigoplus^k\cT Q\times\R^k, h)$ and let $(\bigoplus^k\cT Q\times\R^k, h_\circ)$ be its associated $k$-contact Hamiltonian system. Then, every section $\bar\psi\colon\R^k\to\R^k\times\bigoplus^k\cT Q\times\R^k$ solution to the Hamilton--De Donder--Weyl equations \eqref{eq:HDW-map-Darboux} for the system $(\R^k\times\bigoplus^k\cT Q\times\R^k, h)$ defines a map $\psi\colon\R^k\to\bigoplus^k\cT Q\times\R^k$ solution to the Hamilton--De Donder--Weyl equation \eqref{eq:k-contact-HDW-Darboux-coordinates} for the $k$-contact Hamiltonian system $(\bigoplus^k\cT Q\times\R^k, h_\circ)$, and conversely.
\end{thm}
\begin{proof}
    Since $h = \bar\pi_2^\ast h_\circ$, one has
    \begin{equation}\label{eq:equiv-aut-1}
        \parder{h}{q^i} = \parder{h_\circ}{q^i}\,,\qquad \parder{h}{p_i^\alpha} = \parder{h_\circ}{p_i^\alpha}\,,\qquad \parder{h}{z^\alpha} = \parder{h_\circ}{z^\alpha}\,.
    \end{equation}
    Let $\bar\psi\colon\R^k\to\R^k\times\bigoplus^k\cT Q\times\R^k$ be a section of the projection $\bar\pi_1\colon \R^k\times\bigoplus^k\cT Q\times\R^k\to\R^k$, which in coordinates reads $\bar\psi(t) = (t, \bar\psi^i(t), \bar\psi_i^\alpha(t), \bar\psi^\alpha(t))$ with $t\in\R^k$. We can construct the map $\psi = \bar\pi_2\circ\bar\psi\colon \R^k\to\bigoplus^k\cT Q\times\R^k$, which in coordinates reads $\psi(t) = (\psi^i(t), \psi_i^\alpha(t), \psi^\alpha(t)) = (\bar\psi^i(t), \bar\psi_i^\alpha(t), \bar\psi^\alpha(t))$. Then, if $\bar\psi$ is a solution to the Hamilton--De Donder--Weyl equations \eqref{eq:HDW-map-Darboux}, from \eqref{eq:equiv-aut-1} one obtains that $\psi$ is a solution to the $k$-contact Hamilton--De Donder--Weyl equations \eqref{eq:k-contact-HDW-Darboux-coordinates}.

    Conversely, consider a map $\psi\colon\R^k\to\bigoplus^k\cT Q\times\R^k$. Define $\bar\psi = (\Id_{\R^k},\psi)\colon\R^k\to\R^k\times\bigoplus^k\cT Q\times\R^k$. If $\psi(t) = (\psi^i(t),\psi_i^\alpha(t), \psi^\alpha(t))$, then $\bar\psi(t) = (t, \bar\psi^i(t),\bar\psi_i^\alpha(t), \bar\psi^\alpha(t))$ with $\bar\psi^i(t) = \psi^i(t)$, $\bar\psi_i^\alpha(t) = \psi_i^\alpha(t)$ and $\bar\psi^\alpha(t) = \psi^\alpha(t)$. Note that $\Ima\bar\psi = \graph\psi$. Thus, if $\psi$ is a solution the $k$-contact Hamilton--De Donder--Weyl equations \eqref{eq:k-contact-HDW-Darboux-coordinates}, we have that $\bar\psi$ is a solution to the Hamilton--De Donder--Weyl equations \eqref{eq:HDW-map-Darboux}.
\end{proof}

The following result relates the $k$-vector fields solution to equations \eqref{eq:k-contact-HDW-fields} and \eqref{eq:HDW-field-autonomous}. First, we have to introduce the notion of suspension of a vector field (see \cite[p.\,374]{Abr1978} for the definition of suspension in the context of mechanics).

Let $\bfX = (X_1,\dotsc,X_k)$ be a $k$-vector field on $\bigoplus^k\cT Q\times\R^k$. For every $\alpha = 1,\dotsc,k$ let $\bar X_\alpha\in\X(\R^k\times\bigoplus^k\cT Q\times\R^k)$ be the \textsl{suspension} of the corresponding vector field $X_\alpha$ in $\bigoplus^k\cT Q\times\R^k$ defined as follows: for every $\mathrm{p}\in\bigoplus^k\cT Q\times\R^k$, let $\gamma_{\mathrm{p}}^\alpha\colon\R\to\bigoplus^k\cT Q\times\R^k$ be the integral curve of $X_\alpha$ passing through $\mathrm{p}$. Then, if $x_0 = (x_0^1,\dotsc,x_0^k)\in\R^k$, we can construct the curve $\bar\gamma_{\mathrm{p}}^\alpha\colon\R^k\times\bigoplus^k\cT Q\times\R^k$ passing through the point $\bar{\mathrm{p}} = (x_0,\mathrm{p})\in\R^k\times\bigoplus^k\cT Q\times\R^k$ given by $\bar\gamma_{\bar{\mathrm{p}}}^\alpha(x) = (x_0^1,\dotsc,x_0^\alpha + x, \dotsc, x_0^k; \gamma_{\mathrm{p}}(x))$. Then, $\bar X\in\X(\R^k\times\bigoplus^k\cT Q\times\R^k)$ is the vector field tangent to $\bar\gamma_{\bar{\mathrm{p}}}^\alpha$ at $(x_0,\mathrm{p})$.

In natural coordinates, if $X_\alpha$ has local expression
$$ X_\alpha = A_\alpha^i\parder{}{q^i} + B_{\alpha i}^\beta\parder{}{p_i^\beta} + C_\alpha^\beta\parder{}{z^\beta}\,, $$
one has that $\bar X_\alpha$ is locally given by
$$ \bar X_\alpha = \parder{}{t^\alpha} + \bar A_\alpha^i\parder{}{q^i} + \bar B_{\alpha i}^\beta\parder{}{p_i^\beta} + \bar C_\alpha^\beta\parder{}{z^\beta} = \parder{}{t^\alpha} + \bar\pi_2^\ast(A_\alpha^i)\parder{}{q^i} + \bar\pi_2^\ast(B_{\alpha i}^\beta)\parder{}{p_i^\beta} + \bar\pi_2^\ast(C_\alpha^\beta)\parder{}{z^\beta}\,. $$

\begin{thm}
    Consider an autonomous $k$-cocontact Hamiltonian system $(\R^k\times\bigoplus^k\cT Q\times\R^k, h)$ and let $(\bigoplus^k\cT Q\times\R^k, h_\circ)$ be its associated $k$-contact Hamiltonian system. Then, every $k$-vector field $\bfX\in\X^k(\bigoplus^k\cT Q\times\R^k)$ solution to equations \eqref{eq:k-contact-HDW-fields} defines a $k$-vector field $\bar\bfX\in\X^k(\R^k\times\bigoplus^k\cT Q\times\R^k)$ solution to equations \eqref{eq:HDW-field-autonomous}.

    In addition, $\bfX$ is integrable if and only if its associated $\bar\bfX$ is also integrable.
\end{thm}
\begin{proof}
    Let $\bfX = (X_1,\dotsc,X_k)\in\X^k(\bigoplus^k\cT Q\times\R^k)$ be a solution to equations \eqref{eq:k-contact-HDW-fields}. Define $\bar X_\alpha\in\X(\R^k\times\bigoplus^k\cT Q\times\R^k)$ as the suspension of the corresponding vector field $X_\alpha\in\X(\bigoplus^k\cT Q\times\R^k)$.
    
    Notice that the vector fields $\bar X_\alpha$ are $\bar\pi_2$-projectable, and $(\bar\pi_2)_\ast \bar X_\alpha = X_\alpha$. Thus, we have defined a $k$-vector field $\bar{\bfX}$ in $\R^k\times\bigoplus^k\cT Q\times\R^k$.
    
    Therefore, we have
    \begin{align}
        i_{\bar X_\alpha}\d\bar\eta^\alpha - \d h - (\Lie_{\bar R_\alpha^z}h)\bar\eta^\alpha &= i_{\bar X_\alpha}\d(\bar\pi_2^\ast\eta^\alpha) - \d (\bar\pi_2^\ast h_\circ) - (\Lie_{R_\alpha^z}h_\circ)(\bar\pi_2^\ast\eta^\alpha)\\
        &= \pi_2^\ast\left( i_{(\bar\pi_2)_\ast\bar X_\alpha}\d\eta^\alpha - \d h_\circ - (\Lie_{R_\alpha^z}h)\eta^\alpha \right)\\
        &= \pi_2^\ast\left( i_{X_\alpha}\d\eta^\alpha - \d h_\circ - (\Lie_{R_\alpha^z}h)\eta^\alpha \right)\\
        &= 0\,,
    \end{align}
    since $\bfX = (X_\alpha)$ satisfies equations \eqref{eq:k-contact-HDW-fields}. It is easy to check that the other equations also hold. Therefore, $\bar\bfX = (\bar X_\alpha)$ satisfies equations \eqref{eq:HDW-field-autonomous}.

    In addition, if $\psi\colon\R^k\to\bigoplus^k\cT Q\times\R^k$ is an integral section of $\bfX$, one has that $\bar\psi\colon\R^k\to\R^k\times\bigoplus^k\cT Q\times\R^k$ such that $\bar\psi = (\Id_{\R^k},\psi)$ (see Theorem \ref{thm:equiv-autonomous-maps}) is an integral section of $\bar\bfX$.

    On the other hand, if $\bar\psi$ is an integral section of $\bar\bfX$, equations \eqref{eq:HDW-field-autonomous} hold for the map $\bar\psi(t) = (t, \bar\psi^i(t), \bar\psi_i^\alpha(t), \bar\psi^\alpha(t))$. Since $\bar A_\alpha^i = \bar\pi_2^\ast(A_\alpha^i)$, $\bar B_{\alpha i}^\beta = \pi_2^\ast(B_{\alpha i}^\beta)$ and $\bar C_\alpha^\beta = \bar\pi_2^\ast(C_\alpha^\beta)$, this is equivalent to say that equations \eqref{eq:k-contact-HDW} hold for the map $\psi(t) = (\psi^i(t), \psi_i^\alpha(t), \psi^\alpha(t))$ or, equivalently, $\psi$ is an integral section of $\bfX$.
\end{proof}

Notice that the converse statement of the previous theorem is not true. Actually, the $k$-vector fields that are solutions to the geometric field equations \eqref{eq:HDW-field-autonomous} are not completely determined, and then there are $k$-vector fields in $\R^k\times\bigoplus^k\cT Q\times\R^k$ that are not $\bar\pi_2$-projectable, for instance taking their undetermined components to be not $\bar\pi_2$-projectable. However, if we only consider those solutions which are integral sections of $k$-vector fields solution to the geometric field equations, one can prove that every integrable $k$-vector field $\bar\bfX\in\X^k(\R^k\times\bigoplus^k\cT Q\times\R^k)$ solution to the $k$-cocontact Hamilton--De Donder--Weyl equations is associated with an integrable $k$-vector field $\bfX\in\X^k(\bigoplus^k\cT Q\times\R^k)$ solution to the $k$-contact Hamilton--De Donder--Weyl equations.

The results presented in this section can be translated to the Lagrangian formalism when considering regular autonomous Lagrangians ($\tparder{L}{t^\alpha} = 0$, or equivalently, $\tparder{E_L}{t^\alpha} = 0$).

\section{An example: one-dimensional nonlinear wave equation with damping}\label{sec:7}

A one-dimensional nonlinear wave with an external time-depending forcing can be modeled by the equation
\begin{equation}\label{eq:nonlinear-wave}
    u_{tt} = \frac{\d}{\d x}\left(\parder{f}{u_x}(t, u_x)\right) - \parder{g}{u}(t, u)\,,
\end{equation}
where $u\colon U\subset\R^2\to\R$ and $u(t,x)$, $f(t, u_x)$ and $g(t,u)$ are smooth functions. Notice that if $g(t,u) = 0$ and $f(t,u_x) = c^2u_x^2/2$ with $c\in\R$, we recover the usual wave equation $u_{tt} = c^2u_{xx}$. This equation can be obtained from the Lagrangian function $L\colon \R^2\times\bigoplus^2\T\R\to\R$ \cite{LMM-2008} given by
$$ L(t,x;u,u_t,u_x) = \frac{1}{2}u_t^2 - f(t, u_x) - g(t,u)\,, $$
where we will assume the regularity condition $\dparder{^2f}{u_x^2}\neq 0$. We are going to modify this Lagrangian function in order to add a damping term proportional to $u_t$ to equation \eqref{eq:nonlinear-wave}. 

\subsection*{Lagrangian formalism}

Consider the Lagrangian function with holonomic damping term $\L\colon\R^2\times\oplus^2\T\R\times\R^2$ given by $\L(t,x;u,u_t,u_x;z^t,z^x) = L(t,x;u,u_t,u_x) + \phi(x, z^t)$, where $\phi(x,z^t) = - \gamma(x)z^t$. Then, we have
\begin{equation}\label{eq:Lagrangian-string}
    \L(t,x;u,u_t,u_x;z^t,z^x) = \frac{1}{2}u_t^2 - f(t, u_x) - g(t,u) - \gamma(x)z^t\,.
\end{equation}

For this Lagrangian, we have
\begin{gather}
    \d\L = -\left( \parder{f}{t} + \parder{g}{t} \right)\d t - z^t\parder{\gamma}{x} - \parder{g}{u}\d u + u_t\d u_t - \parder{f}{u_x}\d u_x - \gamma(x)\d z^t\,,\\
    E_\L = \frac{1}{2}u_t^2 - u_x\parder{f}{u_x} + f(t,u_x) + g(t,u) + \gamma(x) z^t\,,\\
    \d E_\L = \left( -u_x\parderr{f}{t}{u_x} + \parder{f}{t} + \parder{g}{t} \right)\d t + \parder{\gamma}{x}z^t\d x + \parder{g}{u}\d u + u_t\d u_t - u_x\parder{^2f}{u_x^2}\d u_x + \gamma(x)\d z^t\,,\\
    \eta_\L^1 = \d z^t - u_t\d u\,,\qquad \d\eta_\L^1 = \d u \wedge \d u_t\,,\\
    \eta_\L^2 = \d z^x + \parder{f}{u_x}\d u \,,\qquad \d\eta_\L^2 = \parderr{f}{t}{u_x}\d t\wedge\d u + \parder{^2f}{u_x^2}\d u_x\wedge\d u\,,\\
    (R^t_\L)_1 = \parder{}{t} - \left(\parder{^2f}{u_x^2}\right)^{-1}\parderr{f}{t}{u_x}\parder{}{u_x}\,,\qquad (R^t_\L)_2 = \parder{}{x}\,,\qquad (R^z_\L)_1 = \parder{}{z^t}\,,\qquad (R^z_\L)_2 = \parder{}{z^x}\,.
\end{gather}
Now, consider a 2-vector field $\bfX = (X_1,X_2)\in\X^2(\R^2\times\bigoplus^2\T\R\times\R^2)$ with local expression
\begin{equation*}
    X_\alpha = A_\alpha^t\parder{}{t} + A_\alpha^x\parder{}{x} + B_\alpha\parder{}{u} + C_{\alpha t}\parder{}{u_t} + C_{\alpha x}\parder{}{u_x} + D_\alpha^t\parder{}{z^t} + D_\alpha^x\parder{}{z^x}\,.
\end{equation*}
For this 2-vector field, the third equation in \eqref{eq:EL-field} gives the conditions $A_1^t = 1$, $A_1^x = 0$, $A_2^t = 0$ and $A_2^x = 1$. We have
$$ i(X_\alpha)\d\eta^\alpha_\L = -B_2\parderr{f}{t}{u_x}\d t + \left( -C_{1t} + A_2^t\parderr{f}{t}{u_x} + C_{2x}\parder{^2f}{u_x^2} \right)\d u + B_1\d u_t - \parder{^2f}{u_x^2}B_2\d u_x\,, $$
and
$$ \d E_\L - (\Lie_{(R^t_\L)_\alpha}E_\L)\d t^\alpha - (\Lie_{(R^z_\L)_\alpha}E_\L)\eta_\L^\alpha = -u_x\parderr{f}{t}{u_x} + \left(\parder{g}{u} + \gamma(x)u_t\right)\d u + u_t\d u_t - u_x\parder{^2f}{u_x^2}\d u_x\,, $$
and then the first equation in \eqref{eq:EL-field} gives the conditions
\begin{align}
    (B_2-u_x)\parderr{f}{t}{u_x} &= 0\,,\label{eq:string-1}\\
    C_{1t} - C_{2x}\parder{^2f}{u_x^2} + \parder{g}{u} + \gamma(x)u_t &= 0\,,\label{eq:string-2}\\
    B_1 &= u_t\,,\label{eq:string-3}\\
    B_2 &= u_x\,.\label{eq:string-4}
\end{align}
Finally, the second equation in \eqref{eq:EL-field} yields $D_1^t + D_2^x = \L$.

Notice that conditions \eqref{eq:string-3} and \eqref{eq:string-4} are the holonomy conditions, while \eqref{eq:string-1} holds identically. Consider now an integral section $\psi(r) = (t(r), x(r); u(r), u_t(r), u_x(r); z^t(r), z^x(r))$ of the 2-vector field $\bfX$. Then, combining equations \eqref{eq:string-3} and \eqref{eq:string-4} into \eqref{eq:string-2}, we obtain the damped nonlinear wave equation:
$$ \parder{^2u}{t^2} - \frac{\d}{\d x}\left(\parder{f}{u_x}(t, u_x)\right) + \parder{g}{u}(t,u) + \gamma(x)\parder{u}{t} = 0\,. $$
In the particular case $f(t,u_x) = c^2u_x/2$, we get
$$ u_{tt} - c^2 u_{xx} + \parder{g}{u}(t,u) + \gamma(x)u_t = 0\,. $$

\subsection*{Hamiltonian formalism}

In order to give a Hamiltonian description of the system introduced above, let us consider the Legendre map associated to the Lagrangian function $\L$ given in \eqref{eq:Lagrangian-string}. The Legendre map associated to $\L$ is the map $\F\L\colon\R^2\times\bigoplus^2\T\R\times\R^2\to\R^2\times\bigoplus^2\cT\R\times\R^2$ given by
$$ \F\L(t,x;u,u_t,u_x;z^t,z^x) = \left(t,x;u,p^t \equiv u_t, p^x\equiv -\parder{f}{u_x};z^t,z^x\right)\,. $$
Notice that the regularity condition $\dparder{^2f}{u_x^2}$ assumed implies that the Legendre map is a local diffeomorphism and thus the Lagrangian $\L$ is regular. In order to simplify the computations, from now on we will consider the particular case $f(t,u_x) = u_x^2/2$.

Consider then the product manifold $\R^2\times\bigoplus^2\cT\R\times\R^2$ equipped with local coordinates $(t, x; u, p^t, p^x; z^t, z^x)$. This manifold has a canonical 2-cocontact structure given by
$$ \tau^1 = \d t\,,\qquad \tau^2 = \d x\,,\qquad \eta^1 = \d z^t - p^t\d u\,,\qquad \eta^2 = \d z^x - p^x\d u\,. $$
It is clear that $\d\eta^1 = \d u\wedge\d p^t$ and $\d\eta^2 = \d u\wedge\d p^x$. In this case, the Reeb vector fields are
$$ R_1^t = \parder{}{t}\,,\qquad R_2^t = \parder{}{x}\,,\qquad R_1^z = \parder{}{z^t}\,,\qquad R_2^z = \parder{}{z^x}\,. $$
The Hamiltonian function $h$ such that $\F\L^\ast h = E_\L$ is
$$ h(t,x;u,p^t,p^x;z^t,z^x) = \frac{1}{2}(p^t)^2 - \frac{1}{2}(p^x)^2 + g(t,u) + \gamma(x)z^t\,. $$
Consider a 2-vector field $\bfY = (Y_1,Y_2)\in\X^2(\R^2\times\bigoplus^2\cT\R\times\R^2)$ with local expression
$$ Y_\alpha = A_\alpha^t\parder{}{t} + A_\alpha^x\parder{}{x} + B_\alpha\parder{}{u} + C_\alpha^t\parder{}{p^t} + C_\alpha^x\parder{}{p^x} + D_\alpha^t\parder{}{z^t} + D_\alpha^x\parder{}{z^x}\,. $$
The Hamilton--De Donder--Weyl equations \eqref{eq:HDW-field} for the 2-vector field $\bfY$ yield the conditions
$$
    \begin{dcases}
        A_1^t = 1\,,\qquad A_1^x = 0\,,\qquad A_2^t = 0\,,\qquad A_2^x = 1\,,\\
        B_1 = p^t\,,\qquad B_2 = -p^x\,,\\
        C_1^t + C_2^x = -\parder{g}{u} - \gamma(x) p^t\,,\\
        D_1^t + D_2^x = \frac{1}{2}(p^t)^2 - \frac{1}{2}(p^x)^2 - g(t,u) - \gamma(x)z^t\,.
    \end{dcases}
$$
Consider now an integral section $\psi(r) = (t(r), x(r); u(r), p^t(r), p^x(r); z^t(r), z^x(r))$ of the 2-vector field $\bfY$. As in the Lagrangian case, it is clear that $\psi$ satisfies the equation
$$ \parder{^2 u}{t^2} - \parder{^2u}{x^2} + \parder{g}{u}(t,u) + \gamma(x)\parder{u}{t} = 0\,, $$
which corresponds to the equation of a damped vibrating string with external forcing.

\section{Conclusions and further research}

In this paper we have introduced a new geometric framework to describe non-autonomous non-conservative field theories: $k$-cocontact structures. This geometric structure combines the notions of $k$-contact and $k$-cosymplectic manifolds and permits to develop Hamiltonian and Lagrangian formulations of non-autonomous non-conservative field theories.

In more detail, in Definition \ref{dfn:k-cocontact-manifold} we have introduced the notion of $k$-cocontact structure as a couple of families of $k$ differential one-forms satisfying certain properties. We have studied the geometry of these manifolds and, in particular, we have proved the existence of Darboux-type coordinates.

Using this geometric framework, the notion of $k$-cocontact Hamiltonian system is presented, along with its corresponding field equations, generalizing the Hamilton--De Donder--Weyl equations of Hamiltonian field theory. We have also compared this formulation with the $k$-contact formalism introduced in \cite{Gas2020} and shown that they are partially equivalent for autonomous field theories.

Moreover, we have developed a Lagrangian formulation for non-autonomous non-conservative field theories. In particular, we have given the conditions determining if a Lagrangian function yields a $k$-cocontact structure and we have introduced the corresponding field equations generalizing the well-known Euler--Lagrange equations.

In order to illustrate the formalisms introduced in this paper, we have studied with full detail the example of a nonlinear damped wave equation with an external time-dependent forcing, both in the Lagrangian and Hamiltonian formulations.

The formalisms introduced in this work open some lines of future research. The first would be to compare the $k$-cocontact formulation introduced in this paper and the $k$-contact formalism \cite{Gas2020,Gas2021} with the so-called multicontact formalism \cite{LGMRR-2022} recently introduced. In this work we have only considered regular Lagrangian functions. The singular case would require the weakening of the notion of $k$-cocontact structure, defining the notion of {\sl $k$-precocontact structure}. Another very interesting line of research would be to study the symmetries of $k$-cocontact systems, obtaining conservation and dissipation laws.

\addcontentsline{toc}{section}{Acknowledgements}
\section*{Acknowledgements}
I acknowledge the financial support of the Ministerio de Ciencia, Innovaci\'on y Universidades (Spain), projects PGC2018-098265-B-C33 and D2021-125515NB-21; and the Novee Idee 2B-POB II project PSP: 501-D111-20-2004310 funded by the ``Inicjatywa Doskonałości - Uczelnia Badawcza'' (IDUB) program.




\bibliographystyle{abbrv}
\bibliography{references.bib}

\begin{thebibliography}{10}

\bibitem{Abr1978}
R.~Abraham and J.~E. Marsden.
\newblock {\em {Foundations of mechanics}}, volume 364 of {\em AMS Chelsea
  publishing}.
\newblock Benjamin/Cummings Pub. Co., New York, 2nd edition, 1978.
\newblock DOI: \href{https://doi.org/10.1090/chel/364}{10.1090/chel/364}.

\bibitem{Arn1989}
V.~I. Arnold.
\newblock {\em {Mathematical Methods of Classical Mechanics}}, volume~60 of
  {\em Graduate Texts in Mathematics}.
\newblock Springer, New York, 2nd edition, 1989.
\newblock DOI:
  \href{https://doi.org/10.1007/978-1-4757-2063-1}{10.1007/978-1-4757-2063-1}.

\bibitem{Awa1992}
A.~Awane.
\newblock {$k$-symplectic structures}.
\newblock {\em J. Math. Phys.}, {\bf 33}(12):4046, 1992.
\newblock DOI: \href{https://doi.org/10.1063/1.529855}{10.1063/1.529855}.

\bibitem{Ban2016}
A.~Banyaga and D.~F. Houenou.
\newblock {\em {A brief introduction to symplectic and contact manifolds}},
  volume~15.
\newblock World Scientific Publishing Co. Pte. Ltd., Singapore, 2016.
\newblock DOI: \href{https://doi.org/10.1142/9667}{10.1142/9667}.

\bibitem{Bra2017a}
A.~Bravetti.
\newblock {Contact Hamiltonian dynamics: The concept and its use}.
\newblock {\em Entropy}, {\bf 10}(19):535, 2017.
\newblock DOI: \href{https://doi.org/10.3390/e19100535}{10.3390/e19100535}.

\bibitem{Bra2018}
A.~Bravetti.
\newblock {Contact geometry and thermodynamics}.
\newblock {\em Int. J. Geom. Methods Mod. Phys.}, {\bf 16}(supp01):1940003,
  2018.
\newblock DOI:
  \href{https://doi.org/10.1142/S0219887819400036}{10.1142/S0219887819400036}.

\bibitem{Bra2017b}
A.~Bravetti, H.~Cruz, and D.~Tapias.
\newblock {Contact Hamiltonian mechanics}.
\newblock {\em Ann. Phys.}, {\bf 376}:17--39, 2017.
\newblock DOI:
  \href{https://doi.org/10.1016/j.aop.2016.11.003}{10.1016/j.aop.2016.11.003}.

\bibitem{Bra2020}
A.~Bravetti, M.~de~León, J.~C. Marrero, and E.~Padrón.
\newblock {Invariant measures for contact Hamiltonian systems: symplectic
  sandwiches with contact bread}.
\newblock {\em J. Phys. A: Math. Theor.}, {\bf 53}:455205, 2020.
\newblock DOI:
  \href{https://doi.org/10.1088/1751-8121/abbaaa}{10.1088/1751-8121/abbaaa}.

\bibitem{Car1991}
J.~F. Cariñena, M.~Crampin, and L.~A. Ibort.
\newblock {On the multisymplectic formalism for first order field theories}.
\newblock {\em Diff. Geom. Appl.}, {\bf 1}(4):345--374, 1991.
\newblock DOI:
  \href{https://doi.org/10.1016/0926-2245(91)90013-Y}{10.1016/0926-2245(91)90013-Y}.

\bibitem{Car2019}
J.~F. Cariñena and P.~Guha.
\newblock {Nonstandard Hamiltonian structures of the Liénard equation and
  contact geometry}.
\newblock {\em Int. J. Geom. Methods Mod. Phys.}, {\bf 16}(supp01):1940001,
  2019.
\newblock DOI:
  \href{https://doi.org/10.1142/S0219887819400012}{10.1142/S0219887819400012}.

\bibitem{Cia2018}
F.~M. Ciaglia, H.~Cruz, and G.~Marmo.
\newblock {Contact manifolds and dissipation, classical and quantum}.
\newblock {\em Ann. Phys.}, {\bf 398}:159--179, 2018.
\newblock DOI:
  \href{https://doi.org/10.1016/j.aop.2018.09.012}{10.1016/j.aop.2018.09.012}.

\bibitem{DeLeo2022}
M.~de~León, J.~Gaset, X.~Gràcia, M.~Muñoz-Lecanda, and X.~Rivas.
\newblock {Time-dependent contact mechanics}.
\newblock {\em Monatsh. Math.}, 2022.
\newblock DOI:
  \href{https://doi.org/10.1007/s00605-022-01767-1}{10.1007/s00605-022-01767-1}.

\bibitem{LGMRR-2022}
M.~de~León, J.~Gaset, M.~C. Muñoz-Lecanda, X.~Rivas, and N.~Román-Roy.
\newblock {Multicontact formalism for non-conservative field theories}.
\newblock arXiv: \href{https://arxiv.org/abs/2209.08918}{2209.08918}, 2022.

\bibitem{DeLeo2019b}
M.~de~León and M.~Lainz-Valcázar.
\newblock {Contact Hamiltonian systems}.
\newblock {\em J. Math. Phys.}, {\bf 60}(10):102902, 2019.
\newblock DOI: \href{https://doi.org/10.1063/1.5096475}{10.1063/1.5096475}.

\bibitem{DeLeo2021b}
M.~de~León, M.~Lainz-Valcázar, and M.~C. Muñoz-Lecanda.
\newblock {The Herglotz Principle and Vakonomic Dynamics}.
\newblock In F.~Nielsen and F.~Barbaresco, editors, {\em Geometric Science of
  Information}, volume 12829 of {\em {Lecture Notes in Computer Science}},
  pages 183--190, Cham, 2021. Springer International Publishing.
\newblock DOI:
  \href{https://doi.org/10.1007/978-3-030-80209-7_21}{10.1007/978-3-030-80209-7\_21}.

\bibitem{DeLeo2021d}
M.~de~León, M.~Lainz-Valcázar, M.~C. Muñoz-Lecanda, and N.~Román-Roy.
\newblock {Constrained Lagrangian dissipative contact dynamics}.
\newblock {\em J. Math. Phys.}, {\bf 62}, 2021.
\newblock DOI: \href{https://doi.org/10.1063/5.0071236}{10.1063/5.0071236}.

\bibitem{LLLR-2022}
M.~de~León, M.~Laínz, A.~López-Gordón, and X.~Rivas.
\newblock {Hamilton--Jacobi theory for contact systems: autonomous and
  non-autonomous}.
\newblock arXiv: \href{https://arxiv.org/abs/2208.07436}{2208.07436}, 2022.

\bibitem{LMM-2008}
M.~de~León, J.~C. Marrero, and D.~{Martín de Diego}.
\newblock {Some Applications of Semi-Discrete Variational Integrators to
  Classical Field Theories}.
\newblock {\em Qual. Theory Dyn. Syst.}, {\bf 7}:195--212, 2008.
\newblock DOI:
  \href{https://doi.org/10.1007/s12346-008-0011-4}{10.1007/s12346-008-0011-4}.

\bibitem{DeLeo1998}
M.~de~León, E.~Merino, J.~A. Oubiña, P.~R. Rodrigues, and M.~Salgado.
\newblock {Hamiltonian systems on $k$-cosymplectic manifolds}.
\newblock {\em J. Math. Phys.}, {\bf 39}(2):876, 1998.
\newblock DOI: \href{https://doi.org/10.1063/1.532358}{10.1063/1.532358}.

\bibitem{DeLeo2001}
M.~de~León, E.~Merino, and M.~Salgado.
\newblock {$k$-cosymplectic manifolds and Lagrangian field theories}.
\newblock {\em J. Math. Phys.}, {\bf 42}(5):2092, 2001.
\newblock DOI: \href{https://doi.org/10.1063/1.1360997}{10.1063/1.1360997}.

\bibitem{DeLeo1989}
M.~de~León and P.~R. Rodrigues.
\newblock {\em {Methods of Differential Geometry in Analytical Mechanics}},
  volume 158 of {\em Mathematics Studies}.
\newblock North-Holland, Amsterdam, 1989.

\bibitem{DeLeo2015}
M.~de~León, M.~Salgado, and S.~Vilariño.
\newblock {\em {Methods of Differential Geometry in Classical Field Theories}}.
\newblock World Scientific, 2015.
\newblock DOI: \href{https://doi.org/10.1142/9693}{10.1142/9693}.

\bibitem{LR-2022}
J.~de~Lucas and X.~Rivas.
\newblock {Contact Lie systems}.
\newblock arXiv: \href{https://arxiv.org/abs/2207.04038}{2207.04038}, 2022.

\bibitem{Gas2020}
J.~Gaset, X.~Gràcia, M.~C. Muñoz-Lecanda, X.~Rivas, and N.~Román-Roy.
\newblock {A contact geometry framework for field theories with dissipation}.
\newblock {\em Ann. Phys.}, {\bf 414}:168092, 2020.
\newblock DOI:
  \href{https://doi.org/10.1016/j.aop.2020.168092}{10.1016/j.aop.2020.168092}.

\bibitem{Gas2019}
J.~Gaset, X.~Gràcia, M.~C. Muñoz-Lecanda, X.~Rivas, and N.~Román-Roy.
\newblock {New contributions to the Hamiltonian and Lagrangian contact
  formalisms for dissipative mechanical systems and their symmetries}.
\newblock {\em Int. J. Geom. Methods Mod. Phys.}, {\bf 17}(6):2050090, 2020.
\newblock DOI:
  \href{https://doi.org/10.1142/S0219887820500905}{10.1142/S0219887820500905}.

\bibitem{Gas2021}
J.~Gaset, X.~Gràcia, M.~C. Muñoz-Lecanda, X.~Rivas, and N.~Román-Roy.
\newblock {A $k$-contact Lagrangian formulation for nonconservative field
  theories}.
\newblock {\em Rep. Math. Phys.}, {\bf 87}(3):347--368, 2021.
\newblock DOI:
  \href{https://doi.org/10.1016/S0034-4877(21)00041-0}{10.1016/S0034-4877(21)00041-0}.

\bibitem{Gei2008}
H.~Geiges.
\newblock {\em {An Introduction to Contact Topology}}, volume 109 of {\em
  Cambridge Studies in Advanced Mathematics}.
\newblock Cambridge University Press, 2008.
\newblock DOI:
  \href{https://doi.org/10.1017/CBO9780511611438}{10.1017/CBO9780511611438}.

\bibitem{Got2016}
S.~Goto.
\newblock {Contact geometric descriptions of vector fields on dually flat
  spaces and their applications in electric circuit models and nonequilibrium
  statistical mechanics}.
\newblock {\em J. Math. Phys.}, {\bf 57}(10):102702, 2016.
\newblock DOI: \href{https://doi.org/10.1063/1.4964751}{10.1063/1.4964751}.

\bibitem{GG22}
K.~Grabowska and J.~Grabowski.
\newblock {A novel approach to contact Hamiltonians and contact
  Hamilton--Jacobi theory}.
\newblock arXiv: \href{https://arxiv.org/abs/2207.04484}{2207.04484}, 2022.

\bibitem{GG22b}
K.~Grabowska and J.~Grabowski.
\newblock {Contact geometric mechanics: the Tulczyjew triples}.
\newblock arXiv: \href{https://arxiv.org/abs/2209.03154}{2209.03154}, 2022.

\bibitem{Gra2000}
X.~Gràcia.
\newblock {Fibre derivatives: Some applications to singular Lagrangians}.
\newblock {\em Rep. Math. Phys.}, {\bf 45}(1):67--84, 2000.
\newblock DOI:
  \href{https://doi.org/10.1016/S0034-4877(00)88872-2}{10.1016/S0034-4877(00)88872-2}.

\bibitem{Gra2021}
X.~Gràcia, X.~Rivas, and N.~Román-Roy.
\newblock {Skinner--Rusk formalism for $k$-contact systems}.
\newblock {\em J. Geom. Phys.}, {\bf 172}:104429, 2022.
\newblock DOI:
  \href{https://doi.org/10.1016/j.geomphys.2021.104429}{10.1016/j.geomphys.2021.104429}.

\bibitem{Kho2013}
A.~L. Kholodenko.
\newblock {\em {Applications of Contact Geometry and Topology in Physics}}.
\newblock World Scientific, 2013.
\newblock DOI: \href{https://doi.org/10.1142/8514}{10.1142/8514}.

\bibitem{Kij1973}
J.~Kijowski.
\newblock {A finite-dimensional canonical formalism in the classical field
  theory}.
\newblock {\em Comm. Math. Phys.}, {\bf 30}(2):99--128, 1973.
\newblock DOI: \href{https://doi.org/10.1007/BF01645975}{10.1007/BF01645975}.

\bibitem{Kij1979}
J.~Kijowski and W.~M. Tulczyjew.
\newblock {\em {A symplectic framework for field theories}}, volume 107 of {\em
  Lecture Notes in Physics}.
\newblock Springer-Verlag, Berlin Heidelberg, 1st edition, 1979.
\newblock DOI:
  \href{https://doi.org/10.1007/3-540-09538-1}{10.1007/3-540-09538-1}.

\bibitem{Lee2013}
J.~M. Lee.
\newblock {\em {Introduction to Smooth Manifolds}}, volume 218 of {\em Graduate
  Texts in Mathematics}.
\newblock Springer New York Heidelberg Dordrecht London, 2nd edition, 2013.
\newblock DOI:
  \href{http://doi.org/10.1007/978-1-4419-9982-5}{10.1007/978-1-4419-9982-5}.

\bibitem{Lib1987}
P.~Libermann and C.-M. Marle.
\newblock {\em {Symplectic Geometry and Analytical Mechanics}}.
\newblock Springer Netherlands, Reidel, Dordretch, oct 1987.
\newblock DOI:
  \href{http://doi.org/10.1007/978-94-009-3807-6}{10.1007/978-94-009-3807-6}.

\bibitem{Liu2018}
Q.~Liu, P.~J. Torres, and C.~Wang.
\newblock {Contact Hamiltonian dynamics: variational principles, invariants,
  completeness and periodic behaviour}.
\newblock {\em Ann. Phys.}, {\bf 395}:26--44, 2018.
\newblock DOI:
  \href{https://doi.org/10.1016/j.aop.2018.04.035}{10.1016/j.aop.2018.04.035}.

\bibitem{Ram2017}
H.~Ramirez, B.~Maschke, and D.~Sbarbaro.
\newblock {Partial stabilization of input-output contact systems on a Legendre
  submanifold}.
\newblock {\em IEEE Transactions on Automatic Control}, {\bf 62}(3):1431--1437,
  2017.
\newblock DOI:
  \href{https://doi.org/10.1109/TAC.2016.2572403}{10.1109/TAC.2016.2572403}.

\bibitem{PhDThesisXRG}
X.~Rivas.
\newblock {\em {Geometrical aspects of contact mechanical systems and field
  theories}}.
\newblock PhD thesis, Universitat Politècnica de Catalunya (UPC), 2021.
\newblock \url{http://hdl.handle.net/10803/673385}.

\bibitem{RiTo-2022}
X.~Rivas and D.~Torres.
\newblock {Lagrangian--Hamiltonian formalism for time-dependent dissipative
  mechanical systems}.
\newblock {\em J. Geom. Mech.}, {\bf 15}(1):1--26, 2022.
\newblock DOI: \href{https://doi.org/10.3934/jgm.2023001}{10.3934/jgm.2023001}.

\bibitem{Rom2011}
N.~Román-Roy, Ángel M.~Rey, M.~Salgado, and S.~Vilariño.
\newblock {On the $k$-symplectic, $k$-cosymplectic and multisymplectic
  formalisms of classical field theories}.
\newblock {\em J. Geom. Mech.}, {\bf 3}(1):113--137, 2011.
\newblock DOI:
  \href{https://doi.org/10.3934/jgm.2011.3.113}{10.3934/jgm.2011.3.113}.

\bibitem{Sim2020}
A.~A. Simoes, M.~de~León, M.~Lainz-Valcázar, and D.~{Martín de Diego}.
\newblock {Contact geometry for simple thermodynamical systems with friction}.
\newblock {\em Proc. R. Soc. A.}, {\bf 476}:20200244, 2020.
\newblock DOI:
  \href{https://doi.org/10.1098/rspa.2020.0244}{10.1098/rspa.2020.0244}.

\end{thebibliography}

\end{document}